\documentclass[journal]{IEEEtran}
\makeatletter
\newcommand*{\rom}[1]{\expandafter\@slowromancap\romannumeral #1@}
\makeatother
\usepackage{amsmath, amssymb, bm, cite, epsfig, psfrag, mathtools, amsthm}

\usepackage{graphicx}
\usepackage{dblfloatfix}

\usepackage{array}
\usepackage{lipsum}
\usepackage{array}
\usepackage{ragged2e}
\usepackage{cite}
\newcolumntype{P}[1]{>{\centering\hspace{0pt}}p{#1}}
\newcolumntype{M}[1]{>{\centering\hspace{0pt}}m{#1}}
\newcolumntype{L}{>{\centering\arraybackslash}m{3cm}}
\usepackage[font=small]{caption}
\usepackage{subcaption}
\usepackage{epstopdf}
\usepackage{longtable}
\usepackage{supertabular,booktabs}
\usepackage{capt-of}
\usepackage{bbm}
\usepackage{multirow}
\usepackage{multicol}
\usepackage[usenames,dvipsnames]{xcolor}
\renewcommand{\arraystretch}{1.5}
\usepackage{etoolbox}
\usepackage{pbox}
\usepackage{fixltx2e}
\usepackage{algorithm}
\usepackage{algpseudocode}
\usepackage{colortbl}
\usepackage{fancyhdr}
\usepackage{atbegshi}
\usepackage{tabu}

\usepackage{bm}
\newtoggle{conference}
\togglefalse{conference} % Use for arxiv version -- also change documentclass
\graphicspath{{figures/}}

\def\diag{\mathop{\mathrm{diag}}}

\newtheorem{theorem}{Theorem}
\newtheorem{lemma}{Lemma}

\fancypagestyle{plain}{
	\fancyhf{}
	\fancyhead[C]{Conference on \LaTeX}             %% C or L or R.
	\fancyfoot[L]{This is a notice}        %        %% C or L or R.

}

\newcommand*{\balancecolsandclearpage}{%
	\close@column@grid
	\cleardoublepage
	\twocolumngrid
}

\allowdisplaybreaks
\begin{document}
%\linespread{1.5}
\makeatletter
\makeatother
\title{A Scalable and Energy Efficient IoT System Supported by Cell-Free Massive MIMO}

\author{Hangsong~Yan,~\IEEEmembership{Student~Member,~IEEE,} Alexei~Ashikhmin,~\IEEEmembership{Fellow,~IEEE,} Hong~Yang
	
\thanks{Hangsong Yan is with NYU Wireless, the Department of Electrical and Computer Engineering, New York University, Brooklyn,
		NY, 11201 USA e-mail: hy942@nyu.edu} % <-this % stops a space
\thanks{Alexei Ashikhmin and Hong Yang are with Nokia Bell Labs, Murray Hill, NJ, 07974 USA e-mail: \{alexei.ashikhmin, h.yang\}@nokia-bell-labs.com}}
\maketitle
%\newpage

\begin{abstract}
An IoT (Internet of things) system supports a massive number of IoT devices wirelessly. We show how to use Cell-Free Massive MIMO (multiple-input and multiple-output) to provide a scalable and energy efficient IoT system.  We employ optimal linear estimation with random pilots to acquire CSI (channel state information) for MIMO precoding and decoding. In the uplink, we employ optimal linear decoder and utilize RM (random matrix) theory to obtain two accurate SINR (signal-to-interference plus noise ratio) approximations involving only large-scale fading coefficients. We derive several max-min type power control algorithms based on both exact SINR expression and RM approximations. Next, we consider the power control problem for downlink (DL) transmission. To avoid solving a time-consuming quasi-concave problem that requires repeat tests for the feasibility of a SOCP (second-order cone programming) problem, we develop a neural network (NN) aided power control algorithm that results in 30 times reduction in computation time. This power control algorithm leads to scalable Cell-Free Massive MIMO networks in which the amount of computations conducted by each AP does not depend on the number of network APs.

Both UL and DL power control algorithms allow visibly improve the system spectral efficiency (SE) and, more importantly, lead to  multi-fold improvements in Energy Efficiency (EE), which is crucial for IoT networks.
\end{abstract}


\section{Introduction}\label{Introduction}
Realizations of internet of things (IoT) is a hot and developing topic for both industry and academia \cite{Chettri2020Compre, Palattella_IoT_2016, Akpakwu_Sruvey_2018}. With the huge benefits of IoT, a whole set of distinctive challenges is also exposed to the wireless physical layer design. These challenges include hyper-connectivity and low latency (the number of wirelessly connected IoT devices will increase exponentially fast and the number of devices served simultaneously is also large), high energy efficiency (EE) and spectral efficiency (SE) (the served devices are in low transmit power regime while having possible high mobility) as well as sporadic transmission.

%\textcolor{blue}{Much work has been done on IoT systems on various aspects such as edge computing~\cite{Abbas2018Mobile, Liu2020Toward}, blockchain for IoT~\cite{Dai2019Blockchain}, efficient resource allocation~\cite{Yang2018Energy, Ansere2020Optimal, Wang2019User}, and wireless powered networks for IoT~\cite{chu2017wireless}.}
Significant work was conducted on physical layer optimization of IoT networks. In particular, in~\cite{Wang2019User} the authors proposed a hybrid nonorthogonal multiple access (NOMA) framework to support the tele-traffic demand of IoT, in~\cite{Yang2018Energy} and \cite{Ansere2020Optimal} the energy efficient resource allocation problems for IoT networks are studied, and in ~\cite{chu2017wireless}, the authors investigated a wirelessly powered sensor network and maximized the system sum throughput via energy beamforming and time allocation designs.

For drastic improvements of IoT network physical layer performance, Massive multiple-input and multiple-output (mMIMO) technique is expected to be one of the best candidates to tackle the challenges faced by IoT systems due to its distinctive advantages such as high SE, high EE, and scalability~\cite{Ngo_2013_Energy, marzetta2016fundamentals}. Some works on cellular mMIMO supported IoT systems are reported and summarized in \cite{bana2019massive} and an analysis on wirelessly powered cell-free (CF) IoT is provided in~\cite{Wang2020Wirelessly} where a joint optimization of uplink (UL) and downlink (DL) power control is also introduced. In \cite{Liu2018Massive1} and \cite{Liu2018Massive2}, the authors provide a framework for user activity detection and channel estimation with a cellular base
station (BS) equipped with a large number of antennas, and characterize the achievable uplink rate. The user activity detection is based on assigning to each user a unique pilot, which serves as the user identifier. These pilots then are used as columns of a sensing matrix. Active users synchronously send their pilots and a BS receives a linear combination of these pilots. Next, the BS runs a compressive sensing detection algorithm for the sensing matrix and identifies the pilots that occur as terms in the linear combination. These pilots, in their turn, reveal the active users. 

In this paper, we propose to support IoT systems with CF mMIMO~\cite{Ngo_17_Cellfree}.
We adopt the pilot assignment method introduced in \cite{Liu2018Massive1} and \cite{Liu2018Massive2} in our work (i.e., each IoT device is assigned a unique pilot) and assume perfect detection of active devices.
The contributions of this work are the following. For UL transmission, linear minimum mean square error (LMMSE) channel estimation and MIMO MMSE receiver are adopted in our design to adapt to IoT scenarios where devices usually have a small transmit power (e.g., 20 mW). Accurate and simple UL signal-to-noise-plus-interference ratio (SINR) approximations incorporating LMMSE channel estimation and MIMO MMSE receiver are derived based on random matrix (RM) theory. Efficient and flexible max-min power control algorithms are designed for both exact and RM SINRs, which can achieve both high SE and EE. To further increase the EE of IoT systems, target rate power control algorithms are designed for both exact and RM SINRs, where a predefined UL common per-device rate can be achieved by all served devices. Simulation results show that the designs proposed can obtain huge SE and EE improvements compared with sub-optimal designs and full power transmission schemes, respectively.

For DL CF mMIMO IoT systems, a neural network (NN) approach is introduced to simplify DL max-min power control. By predicting the normalized transmit power for every access point (AP) under optimal max-min power control, DL max-min power control is converted from a high-complexity quasi-concave problem to a low-complexity convex optimization problem. With the aid of the NN prediction, we further develop a scalable power control algorithm that works for very large areas and has very low complexity. Simulation results show high prediction accuracy of the proposed NN approach and significant EE improvements by the scalable power control algorithm compared with full power transmission schemes.

The organization of this paper is as follows. A CF mMIMO supported IoT system model incorporating the unique pilot assignment and LMMSE channel estimations is given in Section \ref{System Model}.
Section \ref{Uplink Transmission} considers UL transmission under a linear optimal MIMO receiver and two RM SINR approximations are presented. UL Max-min and target rate power control algorithms are introduced in Section \ref{Uplink Power Control} and Section \ref{Uplink Simulation Results} shows the UL simulation results. DL transmission under LMMSE channel estimation and conjugate beamforming (CB) precoding is considered in Section \ref{Downlink Transmission}. Section \ref{Downlink Power Control} shows simplified DL max-min power control and scalable DL power control algorithms which are both aided by NN. DL simulation results are also provided in this Section. The conclusion of the paper is provided in Section \ref{Conclusion}. 

Notation: boldface upper- and lower-case letters denote matrices and column vectors, respectively. $(\cdot)^{T}$ $(\cdot)^{*}$, and $(\cdot)^{H}$ denote transpose, conjugate, and conjugate transpose operations, respectively. $z \sim \mathcal{CN}(0, \sigma^2)$ denotes a circularly symmetric complex Gaussian random variable with zero mean and variance $\sigma^2$. $\mathbf{I}_M$ denotes an $M \times M$ identity matrix, $\text{tr}(\cdot)$ denotes trace operator of a matrix, and $\|\cdot\|_2$ denotes the Euclidean norm operation. $\diag\{\mathbf{v}\}$ denotes generating a square diagonal matrix with the elements of vector $\mathbf{v}$ on the main diagonal and $\diag\{\mathbf{A}\}$ denotes generating a column vector of the main diagonal elements of $\mathbf{A}$.

\section{System Model}\label{System Model}
We consider an IoT system supported by CF mMIMO. $M$ APs are uniformly distributed in a wide serving area, they cooperate with each other to serve $\bar{K}$ IoT devices. We denote the number of active devices at one moment as $K$ and assume that the $K$ active devices are randomly chosen from all $\bar{K}$ IoT devices with $M >> K$. Note that $\bar{K}$ is the total number of IoT devices served in this area and is much larger than $K$ (i.e., $\bar{K} >> K$). In our system, orthogonal frequency-division multiplexing (OFDM) is used and we assume a flat-fading channel model for each OFDM subcarrier. For a given subcarrier, the {\em channel coefficient} $g_{mk}$ between the $m$-th AP and the $k$-th device is modeled as:
\begin{equation}
	\label{eq:ch_coeffi}
	g_{mk} = \sqrt{\beta_{mk}}h_{mk},
\end{equation}
where $\beta_{mk}, m = 1,..., M, k = 1,..., K$ are the {\em large-scale fading coefficients} which include path loss and shadow fading. $h_{mk}, m = 1,..., M, k = 1,..., K$ are {\em small-scale fading coefficients} with $i.i.d.\ \mathcal{CN}(0,1)$ and stay constant during the coherence interval of length $\tau_c$ measured in the number of OFDM symbols. As in \cite{Liu2018Massive1} and  \cite{ Liu2018Massive2}, we assign a unique pilot $\boldsymbol{\psi}_k$ to each IoT device. During the pilot transmission, pilots are synchronously transmitted by active devices with the {\em pilot length} being $\tau$. Denote $\sqrt{\tau}\boldsymbol{\psi}_{k} \in \mathbb{C}^{\tau\times 1}$ as the pilot transmitted by the $k$-th device, where $||\boldsymbol{\psi}_{k}||_{2} = 1$, the received pilot signal at AP sides are given as:
\begin{equation}
	\label{eq:re_pilot}
	\mathbf{Y} = [\mathbf{y}_{1}\,\mathbf{y}_{2}\, ...\, \mathbf{y}_{m}] = \sqrt{\tau\rho_{p}}\boldsymbol{\Psi}\mathbf{G}^{T} + \mathbf{W},
\end{equation}
\begin{equation}
	\mathbf{y}_{m} = \sqrt{\tau\rho_{p}}\boldsymbol{\Psi}\mathbf{g}_{[m]} + \mathbf{w}_{m},
\end{equation}
where $\rho_{p}$ is the normalized {\em pilot signal-to-noise ratio} (SNR) of each pilot symbol, $\boldsymbol{\Psi} = [\boldsymbol{\psi}_{1}\,\boldsymbol{\psi}_{2}\,...\,\boldsymbol{\psi}_{K}] \in \mathbb{C}^{\tau \times K}$ are {\em pilot matrix} for $K$ devices, $\mathbf{G}^{T} = [\mathbf{g}_{[1]}\,\mathbf{g}_{[2]}\,...\,\mathbf{g}_{[M]}] \in \mathbb{C}^{K\times M}$ are the {\em channel coefficient matrix}, here $\mathbf{g}_{[m]} = [g_{m1},g_{m2},...,g_{mK}]^{T} \in \mathbb{C}^{K\times 1}$. $\mathbf{W}\in \mathbb{C}^{\tau \times M}$ are noise matrix with $i.i.d.\ \mathcal{CN}(0,1)$ components. Optimal channel estimation (i.e., LMMSE) are applied at the AP sides and the estimated channel coefficients at the $m$-th AP have the form:
\begin{equation}
\begin{aligned}
	\label{eq:est_g}
	\hat{\bf g}_{[m]} = \sqrt{\tau\rho_{p}}\mathbf{B}_{m}\boldsymbol{\Psi}^{H}\left(\tau\rho_{p}\boldsymbol{\Psi}\mathbf{B}_{m}\boldsymbol{\Psi}^{H} + \mathbf{I}_{\tau}\right)^{-1}\mathbf{y}_{m},\\
	\hat{g}_{mk} = \sqrt{\tau\rho_{p}}\beta_{mk}\boldsymbol{\psi}_{k}^{H}\left(\tau\rho_{p}\boldsymbol{\Psi}\mathbf{B}_{m}\boldsymbol{\Psi}^{H} + \mathbf{I}_{\tau}\right)^{-1}\mathbf{y}_{m},
\end{aligned}
\end{equation}
where $\mathbf{B}_{m} = \diag\{[\beta_{m1},\beta_{m2},...,\beta_{mK}]\}$. Then the variance of the estimated channel coefficient $\hat{g}_{mk}$ is equal to
\begin{equation}
\begin{aligned}
	\label{eq:gamma}
	\gamma_{mk} \triangleq \mathbb{E}\{|\hat{g}_{mk}|^{2}\} = 
	\sqrt{\tau\rho_{p}}\beta_{mk}\boldsymbol{\psi}_{k}^{H}\mathbf{a}_{mk},
\end{aligned}
\end{equation}
where $\mathbf{a}_{mk} = \sqrt{\tau\rho_{p}}\beta_{mk}\left(\tau\rho_{p}\boldsymbol{\Psi}\mathbf{B}_{m}\boldsymbol{\Psi}^{H} + \mathbf{I}_{\tau}\right)^{-1}\boldsymbol{\psi}_{k}$. 
Let $\tilde{g}_{mk}={ g}_{mk}-\hat{g}_{mk}$ be the channel estimation error. Since the estimation error and the estimate are orthogonal under LMMSE estimation, the variance of $\tilde{g}_{mk}$ is given by
\begin{equation}
	\label{eq:est_err_g}
	\mathbb{E}\{|\tilde{g}_{mk}|^{2}\} = \beta_{mk} - \gamma_{mk}.
\end{equation} 
 
\section{Uplink Transmission}\label{Uplink Transmission}
%In this section, we consider UL transmission.

\subsection{Uplink Data Transmission}\label{sec:UL_Data_Trans}
Define $\eta_{k}$ as the {\em UL power coefficient} for the $k$-th device.
%, which has constraint $0 \leq \eta_{k} \leq 1$. 
The signals received by the APs are
\begin{equation}
	\label{eq:rece_sig_UL}
	\mathbf{y}_{u} = \sqrt{\rho_{u}}\sum_{k=1}^{K}\sqrt{\eta_{k}}\mathbf{g}_{k}s_{k} + \mathbf{w}_{u},
\end{equation}
where $s_{k}$ is the {\em data symbol} transmitted by the $k$-th device, which satisfies $\mathbb{E}\{|s_k|^2\} = 1$, $\rho_{u}$ is the {\em normalized UL SNR}, $\mathbf{g}_{k} = [g_{1k}, g_{2k},..., g_{Mk}]^{T} \in \mathbb{C}^{M\times 1}$ is the channel vector between the $k$-th device and all APs, and $\mathbf{w}_{u} \in \mathbb{C}^{M\times 1}$ is the noise vector with $i.i.d.\ \mathcal{CN}(0,1)$ components.

We assume that APs cooperate to estimate $s_k$ by using a linear MIMO receiver, $\mathbf{v}_k$, and computing:
\begin{equation}
\label{eq:decode_symbol}
\begin{aligned}
	\hat{s}_{k} = \mathbf{v}_{k}^{H}\mathbf{y}_{u} =  \mathbf{v}_{k}^{H}\bigg(&\sqrt{\rho_{u}\eta_{k}}\hat{\mathbf{g}}_{k}s_{k} + \sqrt{\rho_{u}}\sum_{k'\neq k}\sqrt{\eta_{k'}}\hat{\mathbf{g}}_{k'}s_{k'} +\\ &\sqrt{\rho_{u}}\sum_{k'=1}^{K}\sqrt{\eta_{k'}}\tilde{\mathbf{g}}_{k'}s_{k'} + \mathbf{w}_{u}\bigg).
\end{aligned}
\end{equation}

Based on (\ref{eq:est_err_g}) and (\ref{eq:decode_symbol}), the UL SINR expression for the $k$-th data symbol is given as:
\begin{equation}
	\label{eq:SINR_k_UL}
	\text{SINR}_{k}^{u}(\boldsymbol{\eta}) = \frac{\rho_{u}\eta_{k}\mathbf{v}_{k}^{H}\hat{\mathbf{g}}_{k}\hat{\mathbf{g}}_{k}^{H}\mathbf{v}_{k}}{\mathbf{v}_{k}^{H}\left(\rho_{u}\sum_{k'\neq k}\eta_{k'}\hat{\mathbf{g}}_{k'}\hat{\mathbf{g}}_{k'}^{H} + \mathbf{D}\right)\mathbf{v}_{k}}, 
\end{equation}
where 
\begin{equation}
	\label{eq:def_D}
	\mathbf{D} = \rho_{u}\sum_{k'=1}^{K}\eta_{k'}(\mathbf{B}_{k'}-\boldsymbol{\Gamma}_{k'}) + \mathbf{I}_{M},
\end{equation}
$\mathbf{B}_{k'} \triangleq \diag\{[\beta_{1k'}, \beta_{2k'},..., \beta_{Mk'}]\}$, and $\mathbf{\Gamma}_{k'} \triangleq \diag\{[\gamma_{1k'}, \gamma_{2k'},...,\gamma_{Mk'}]\}$. Based on (\ref{eq:decode_symbol}) and (\ref{eq:SINR_k_UL}), and using Rayleigh-Ritz theorem, we find the optimal (LMMSE) choice of ${\bf v}_k$:
\begin{equation}
	\label{eq:decoder_MMSE}
	\mathbf{v}_{k}^{\text{MMSE}} = \sqrt{\rho_{u}\eta_{k}}\left(\rho_{u}\sum_{k'=1}^{K}\eta_{k'}\hat{\mathbf{g}}_{k'}\hat{\mathbf{g}}_{k'}^{H} + \mathbf{D}\right)^{-1}\hat{\mathbf{g}}_{k}.
\end{equation}
Substituting (\ref{eq:decoder_MMSE}) into (\ref{eq:SINR_k_UL}), we then obtain the corresponding SINR expression:
\begin{equation}
\begin{aligned}
	\label{eq:SINR_k_UL_MMSE}
	&\text{SINR}_{k}^{u, \text{MMSE}}(\boldsymbol{\eta}) = \\ &\frac{\rho_{u}\eta_{k}\hat{\mathbf{g}}_{k}^{H}\left(\rho_{u}\sum_{k'=1}^{K}\eta_{k'}\hat{\mathbf{g}}_{k'}\hat{\mathbf{g}}_{k'}^{H} + \mathbf{D}\right)^{-1}\hat{\mathbf{g}}_{k}}{1-\rho_{u}\eta_{k}\hat{\mathbf{g}}_{k}^{H}\left(\rho_{u}\sum_{k'=1}^{K}\eta_{k'}\hat{\mathbf{g}}_{k'}\hat{\mathbf{g}}_{k'}^{H} + \mathbf{D}\right)^{-1}\hat{\mathbf{g}}_{k}}.
\end{aligned}
\end{equation}

\subsection{RM based SINR Approximations}
In this section, two SINR approximations of (\ref{eq:SINR_k_UL_MMSE}) are derived based on random matrix (RM) theory \cite{Hoydis_13_UL/DL}, \cite{Wagner_12_MISO}. 
\subsubsection{RM Approximation 1}
\begin{equation}
	\label{eq:RMT_Approx1}
	\text{SINR}_{k}^{u,\text{AP1}}(\boldsymbol{\eta}) = \frac{\rho_{u}\eta_{k}\text{tr}\left(\boldsymbol{\Gamma}_{k}\mathbf{T}\right)}{M},
\end{equation}
where
\begin{equation}
\label{eq:T_Approx1}
\mathbf{T} = \left(\frac{\rho_{u}}{M}\sum_{k'=1}^{K}\frac{\eta_{k'}\mathbf{\Gamma}_{k'}}{1 + e_{k'}} + \frac{\mathbf{D}}{M}\right)^{-1},
\end{equation}
\begin{equation}
\label{eq:e_k}
e_{k'} = \lim_{t\rightarrow \infty} e_{k'}^{(t)}\ \text{with}\ e_{k'}^{(0)} = M, \;\;\forall k',
\end{equation}
\begin{equation}
\label{eq:e_k_t}
e_{k'}^{(t)} = \frac{\rho_{u}\eta_{k'}}{M}\text{tr}\left(\mathbf{\Gamma}_{k'}\Bigg(\frac{\rho_{u}}{M}\sum_{j=1}^{K}\frac{\eta_{j}\mathbf{\Gamma}_{j}}{1 + e_{j}^{(t-1)}} + \frac{\mathbf{D}}{M}\Bigg)^{-1}\right), \forall k'.
\end{equation}

\subsubsection{RM Approximation 2}
\begin{equation}
\label{eq:RMT_Approx2}
\text{SINR}_{k}^{u,\text{AP2}}(\boldsymbol{\eta}) = \frac{\rho_{u}\eta_{k}\text{tr}\left(\mathbf{\Gamma}_{k}\mathbf{T}_{k}\right)}{M},
\end{equation}
where
\begin{equation}
\label{eq:T_k}
\mathbf{T}_{k} = \left(\frac{\rho_{u}}{M}\sum_{k'\neq k}^{K}\frac{\eta_{k'}\mathbf{\Gamma}_{k'}}{1 + e_{k,k'}} + \frac{\mathbf{D}}{M}\right)^{-1},
\end{equation}
\begin{equation}
e_{k,k'} = \lim_{t\rightarrow \infty}e_{k,k'}^{(t)}\ \text{with}\ e_{k,k'}^{(0)} = M, \;\; \forall k',
\end{equation}
\begin{equation}
e_{k,k'}^{(t)}=\frac{\rho_{u}\eta_{k'}}{M}\text{tr}\left(\mathbf{\Gamma}_{k'}\Bigg(\frac{\rho_{u}}{M}\sum_{j\neq k}^{K}\frac{\eta_{j}\mathbf{\Gamma}_{j}}{1 + e_{k,j}^{(t-1)}} + \frac{\mathbf{D}}{M}\Bigg)^{-1}\right), \forall k'.
\end{equation}
The derivations of RM Approximation 1 and 2 are given in Appendix \ref{Proof for RM Approximations}. It is noted that both RM Approximation 1 and 2 involve only large-scale fading coefficients. It is also important to note that compared with the RM approximation derived in \cite{nayebi2016MMSE} for the case of reuse of orthogonal pilots, RM Approximation 1 and 2 have much simpler form, which allows for low complexity and infrequent power control algorithms presented in the next section.   

\section{Uplink Power Control}\label{Uplink Power Control}
In this section, we consider UL power control with per-device power constraint. The UL power coefficients defined in Section \ref{sec:UL_Data_Trans} should satisfy the following power constraint: $0 \leq \mathbb{E}\{|x_k|^2\} \leq \rho_u, k = 1,2,...,K$ where $x_k = \sqrt{\rho_u}\eta_ks_k$ is the signal transmitted by the $k$-th device. Since $\mathbb{E}\{|s_k|^2\} = 1$, the UL power constraint can be rewritten as $0 \leq \eta_k \leq 1, \forall k$.

\subsection{Max-min power control by exact SINR}\label{sec:max_min_exact_SINR} Based on the power constraint given above, basic UL max-min power control with random pilots under an IoT system can be formulated as the problem shown below:
\begin{equation}
	\label{eq:max_min_prbm_UL}
	\begin{aligned}
		&\max_{\boldsymbol{\eta}}\min_{k=1,...,K}\text{SINR}^{u}_{k}\\
		&\text{s.t.}\; 0 \leq \eta_{k} \leq 1,\, k = 1,..., K,
	\end{aligned}
\end{equation}
where $\text{SINR}^u_k$ in (\ref{eq:max_min_prbm_UL}) can be replaced by both $\text{SINR}^{u,\text{MMSE}}_k$ given in (\ref{eq:SINR_k_UL_MMSE}) and the RM SINR Approximations.  
In this section, we introduce an iterative weighted max-min power control algorithm designed based on exact SINR given in (\ref{eq:SINR_k_UL_MMSE}). In the algorithm, a {\em rate weighting} vector $\mathbf{u} = [u_{1}, u_{2},..., u_{K}]^{T} \in \mathbb{R}_{0^{+}}^{K}$ constrained by $\|\mathbf{u}\|_{2} = 1$ 
is incorporated. Vector $\mathbf{u}$ can be used to drop some devices under poor channel condition by assigning very small weights to these devices. If all the devices are required to achieve the same rate, $u_{k} = 1/\sqrt{K},\ \forall k$. On the other hand, a {\em power weighting} vector $\boldsymbol{\nu} = [\nu_{1}, \nu_{2}, ..., \nu_{K}]^{T} \in \mathbb{R}_{0^{+}}^{K}$ is also included and the weighted normalized maximum transmit power of the $k$-th device is defined as $\rho'_{u,k} \triangleq \rho_{u}\nu_{k}$. From (\ref{eq:SINR_k_UL_MMSE}), we observed that the task of achieving a uniform data rate for all served devices was equivalent for the term, $\rho_{u}\eta_{k}\hat{\mathbf{g}}_{k}^{H}\left(\rho_{u}\sum_{k'=1}^{K}\eta_{k'}\hat{\mathbf{g}}_{k'}\hat{\mathbf{g}}_{k'}^{H} + \mathbf{D}\right)^{-1}\hat{\mathbf{g}}_{k}$, in (\ref{eq:SINR_k_UL_MMSE}) having the same value for every served device. It is also observed that the term, $\rho_{u}\sum_{k'=1}^{K}\eta_{k'}\hat{\mathbf{g}}_{k'}\hat{\mathbf{g}}_{k'}^{H} + \mathbf{D}$, is included in the SINR expression for each active device (i.e., independent with the device index). With these two observations, an iterative max-min power control algorithm based on exact SINR can be designed.
Here, we define matrix $\mathbf{J}_k,\ k = 1, 2,..., K$ as
\begin{equation}
	\label{eq:J_matrix}
	\mathbf{J}_{k} \triangleq \hat{\mathbf{g}}_{k}\hat{\mathbf{g}}_{k}^{H} + \mathbf{B}_{k} - \boldsymbol{\Gamma}_{k}.
\end{equation}
The details of the algorithm incorporating vectors $\mathbf{u}$ and $\boldsymbol{\nu}$ are given below in {\bf Algorithm} \ref{Algorithm 1} where (\ref{eq:eta_1}) guarantees that the power constraints are satisfied. 
\begin{algorithm}
	\caption{Max-min Power Control - Exact SINR}
	\label{Algorithm 1}
	\begin{algorithmic}[1]
		\item Initialize vectors $\mathbf{u}$ and $\boldsymbol{\nu}$ with predefined setting. Initialize $\eta_{k}^{(0)} = 1, \forall k$, and $d_{k}^{(0)}, \forall k$ as
		\begin{equation}
		\label{eq:d_k_0}
		d_{k}^{(0)} = \rho'_{u,k}\hat{\mathbf{g}}_{k}^{H}\left(\sum_{k'=1}^{K}\rho'_{u,k'}\eta_{k'}^{(0)}\mathbf{J}_{k'} + \mathbf{I}_{M}\right)^{-1}\hat{\mathbf{g}}_{k}.
		\end{equation}
		Set $n = 0$ and choose a tolerance $\epsilon > 0$.
		\item Compute $\alpha = \min_{k'}(d_{k'}^{(n)}/u_{k'}), k' = 1,..., K$ and Update $d_{k}^{(n+1)}, \forall k$ as below
		\begin{equation}
		\begin{aligned}
		\label{eq:d_updation}
		d_{k}^{(n+1)} = \rho'_{u,k}\hat{\mathbf{g}}_{k}^{H}\left(\alpha\sum_{k'=1}^{K}\frac{\rho'_{u,k'}u_{k'}}{d_{k'}^{(n)}}\mathbf{J}_{k'} + \mathbf{I}_{M}\right)^{-1}\hat{\mathbf{g}}_{k}.
		\end{aligned}
		\end{equation}
		\item Stop if $\max_{k}|d_{k}^{(n+1)} - d_{k}^{(n)}|\leq \epsilon, k = 1,..., K$, set $d_{k} = d_{k}^{(n+1)}, \forall k$, and the power control coefficients, $\eta_{k}, \forall k$ are given by 
		\begin{equation}
		\label{eq:eta_1}
		\eta_{k} = \frac{\min_{k'}\left(d_{k'}/u_{k'}\right)}{d_{k}/u_{k}}, \,  k' = 1,...,K.
		\end{equation}
		Otherwise, set $n = n + 1$ and go to step 2.
	\end{algorithmic}
\end{algorithm}
\begin{theorem}\label{converg_algm1}
	Given $\{\hat{g}_{mk}\}$, $\{\beta_{mk}\}$, $\{\gamma_{mk}\}$, $\rho_u$, $\mathbf{u}$, and $\boldsymbol{\nu}$, if a max-min power control is feasible, the iteration in Step 2 of Algorithm {\ref{Algorithm 1}} converges to a unique solution. The $\{\eta_{k}\}$ given by step 3 of Algorithm \ref{Algorithm 1} realize this max-min power control.
\end{theorem}
\begin{proof}
	The proof of Theorem \ref{converg_algm1} is given in Appendix \ref{Convergence of Algorithm 1 and 3}.
\end{proof}

Note that each element of vector $\mathbf{u}$ can have different values. As mentioned earlier, when we would like to drop some devices under poor channel condition, vector $\mathbf{u}$ can be designed as follows. {\bf Algorithm} \ref{Algorithm 1} will be firstly run with $u_k=1/\sqrt{K},\ \forall k$, to obtain $\eta_{k}, \forall k$. We find $K_p$ devices with the largest power coefficients. These devices consume most of the energy and drag the rates of other devices down. Thus it looks natural to lower their rates by assigning their $u_k$ to some small value $u_p$. The value of $u_{k}$ for the remaining devices is set as $u_{g} > u_{p}$. The relationship between $u_{g}$ and $u_{p}$ is given by
\begin{equation}
\label{eq:relation_up_ug}
u_{p}^{2}K_{p} + u_{g}^{2}\left(K- K_{p}\right) = 1.
\end{equation} 
If we would like to virtually drop $K_p$ devices, we set $u_p$ as a very small number (e.g. $10^{-8}$). 

\subsection{Max-min power control by RM Approximation 1}
For RM Approximation 1 given in (\ref{eq:RMT_Approx1}) - (\ref{eq:e_k_t}), matrix $\mathbf{T}$ is included in the SINR expression for every active device and is independent with the device index, thus a max-min power control algorithm based on RM Approximation 1 can be designed.
Its detailed form incorporating the rate and power weighting vectors is given in {\bf Algorithm} \ref{Algorithm 2} below.
\begin{algorithm}
	\caption{Max-min Power Control - RM SINR}
	\label{Algorithm 2}
	\begin{algorithmic}[1]
	\item Initialize $\mathbf{u}$, $\boldsymbol{\nu}$ as predetermined, $\eta_{k}^{(0)} = 1, \forall k$, $\mathbf{D}^{(0)} = \sum_{k'=1}^{K}\rho'_{u,k'}\eta_{k'}^{(0)}(\mathbf{B}_{k'}-\boldsymbol{\Gamma}_{k'}) + \mathbf{I}_{M}$. Initialize $\mathbf{T}^{(0)} = \left(\frac{1}{M}\sum_{k'=1}^{K}\frac{\eta_{k'}^{(0)}\rho'_{u,k'}\mathbf{\Gamma}_{k'}}{1+e_{k'}} + \frac{\mathbf{D}^{(0)}}{M}\right)^{-1}$ where $e_{k'}, \forall k'$ are computed by (\ref{eq:e_k}) and (\ref{eq:e_k_v_t}) below
	\begin{equation}
	\label{eq:e_k_v_t}
		e_{k'}^{(t)} =  \frac{\rho'_{u,k'}\eta_{k'}^{(0)}}{M}\text{tr}\Bigg(\mathbf{\Gamma}_{k'}\Bigg(\frac{1}{M}\sum_{j=1}^{K}\frac{\eta_{j}^{(0)}\rho'_{u,j}\mathbf{\Gamma}_{j}}{1 + e_{j}^{(t-1)}} + \frac{\mathbf{D}^{(0)}}{M}\Bigg)^{-1}\Bigg)
	\end{equation}
	Set $n = 0$ and choose a tolerance $\epsilon > 0$.
	
	\item Compute $\alpha = \min_{k'}\text{tr}\left(\nu_{k'}\mathbf{\Gamma}_{k'}\mathbf{T}^{(n)}\right)/u_{k'},\ k' = 1, ..., K$, Update $\mathbf{T}^{(n+1)}$ as
	\begin{equation}
	\begin{aligned}
	\label{eq:T_update}
	&\mathbf{T}^{(n+1)} = \\ &\left(\frac{\alpha}{M}\sum_{k'=1}^{K}\frac{\rho_{u}u_{k'}}{\text{tr}\left(\mathbf{\Gamma}_{k'}\mathbf{T}^{(n)}\right)}\left(\mathbf{B}_{k'} - \frac{\xi_{k'}\mathbf{\Gamma}_{k'}}{1 + \xi_{k'}}\right) + \frac{\mathbf{I}_{M}}{M}\right)^{-1},
	\end{aligned} 
	\end{equation}
	where $\xi_{k'} = \frac{\rho_{u}\alpha u_{k'}}{M}$.
	
	\item Stop if $\|\mathbf{T}^{(n+1)} - \mathbf{T}^{(n)}\|_{2} \leq \epsilon$. Set $\mathbf{T} = \mathbf{T}^{(n+1)}$ and the power control coefficients, $\eta_{k}, \forall k$ are given by
	\begin{equation}
	\label{eq:eta_al_AP1}
	\eta_{k} = \frac{\min_{k'}\left(\text{tr}\left(\nu_{k'}\mathbf{\Gamma}_{k'}\mathbf{T}\right)/u_{k'}\right)}{\text{tr}\left(\nu_{k}\mathbf{\Gamma}_{k}\mathbf{T}\right)/u_{k}},\ k' = 1,..., K.
	\end{equation}
	Otherwise, set $n = n + 1$ and go to step 2.
	\end{algorithmic}	
\end{algorithm}
\begin{theorem}\label{theorem: converg_algm2} Given $\{\beta_{mk}\}$, $\{\gamma_{mk}\}$, $\rho_u$, $\mathbf{u}$, and $\boldsymbol{\nu}$, if a max-min power control is feasible, the iteration in Step 2 of Algorithm \ref{Algorithm 2} converges to a unique solution. The $\{\eta_{k}\}$ given by step 3 of Algorithm \ref{Algorithm 2} realize this max-min power control.
\end{theorem}
\begin{proof}
	The proof of Theorem \ref{theorem: converg_algm2} is given in Appendix \ref{Convergence of Algorithgm 2 and 4}.
\end{proof}

\subsection{Power Control with Target Rate by Exact SINR}\label{sec:target_exact_SINR}
Under IoT systems, the IoT devices will use energy harvesting and/or infrequently replaced batteries. Thus, high EE of the designed system is highly desirable to support IoT systems. Here we define the {\em UL EE} of a system as
\begin{equation}
\label{eq:EE}
E_{u} \triangleq \frac{\sum_{k = 1}^{K}R_{k}^{u}}{P_{u}\sum_{k=1}^{K}\eta_{k}},
\end{equation}
where $R_{k}^{u}$ is the UL rate for the $k$-th device, $P_{u}$ is the UL maximum transmit power per data symbol.

Since in some IoT application scenarios, the data rate requirement for each device is not relatively high, a target rate power control which can achieve a predetermined target rate for every served device while keeping high EE is desired. A basic UL target rate power control problem with random pilots under an IoT system can be formulated as shown below: 
\begin{equation}
	\label{eq:target_probm}
	\begin{aligned}
		&\text{Find} \;\eta_k, k = 1,2,...,K\\
		&\text{s.t.}\; \text{SINR}_k^u = S_t,\, k = 1, 2, ..., K,\\
		& \quad\;\; 0 \leq \eta_k \leq 1,\;  k = 1,2, ..., K,	
	\end{aligned}
\end{equation}
where $S_t$ is the predetermined target SINR value and $\text{SINR}_k^u$ can also be replaced by $\text{SINR}_k^{u, \text{MMSE}}$ or RM SINR Approximations. Note that max-min power control can be regarded as a kind of target rate power control where the target rate is the maximum uniform data rate achieved by all served devices. Thus, the design approach for the max-min power control algorithm mentioned in Section \ref{sec:max_min_exact_SINR} can be applied for the target rate power control. The difference is that for max-min power control, $\alpha$ given in step 2 of Algorithm \ref{Algorithm 1} needs to be updated in iteration so that $\alpha$ will go towards the target value where all served devices can achieve the maximum uniform data rate, while in target rate power control, $\alpha$ is a fixed value dependent on $S_t$. 

In this section, we introduce a target rate power control algorithm incorporating rate and power weighting vectors. Before going to the details of the algorithm, we first compute the per-device rate under full power condition using (\ref{eq:SINR_k_UL_MMSE}) for all devices. We regard the devices as {\em poor devices} if their per-device rates under full power are smaller than the target rate and regard the remaining devices as {\em good devices}. We set $u_k = u_p$ and $u_k = u_g$ for poor and good devices, respectively. The algorithm details are then given in {\bf Algorithm} \ref{Algorithm 3}.
\begin{algorithm}
	\caption{Target Rate Power Control - Exact SINR}
	\label{Algorithm 3}
	\begin{algorithmic}[1]
		\item Initialize $\boldsymbol{\nu}$, $\eta_{k}^{(0)}, \forall k$, and $d_{k}^{(0)}, \forall k$ as in step 1 of Algorithm \ref{Algorithm 1}. With a target SINR, $S_{t}$, compute $\alpha = S_{t}/(1 + S_{t})$. Set $n = 0$ and choose a tolerance $\epsilon > 0$.
		
		\item Update $d_{k}^{(n+1)}, \forall k$ using (\ref{eq:d_updation}) where $\rho'_{u,k'}u_{k'}, \forall k'$ in (\ref{eq:d_updation}) are substituted by $\rho'_{u,k'}, \forall k'$.
		
		\item Stop if $\max_{k}|d_{k}^{(n+1)}-d_{k}^{(n)}| \leq \epsilon, k = 1,..., K$ and set $d_{k} = d_{k}^{(n+1)}, \forall k$. Otherwise, set $n = n + 1$ and go to step 2.
		
		\item The power control coefficients are computed as $\eta_{k} = \alpha/d_{k}, \forall k$. If the constraints $0 \leq \eta_{k} \leq 1, \forall k$ are satisfied, then Algorithm \ref{Algorithm 3} ends. Otherwise, initialize $\eta_{k}^{(0)}$, and $d_{k}^{(0)}, \forall k$ as in step 1 of Algorithm \ref{Algorithm 1}. Assign the value of each element of vector $\mathbf{u}$ by $u_{g}$ or $u_{p}$ according to the per-device rate under full power case. Set $n = 0$, $\alpha = \alpha/u_{g}$, and go to step 5.
		
		\item Update $d_{k}^{(n+1)}, \forall k$ using (\ref{eq:d_updation}).
		
		\item Stop if $\max_{k}|d_{k}^{(n+1)} - d_{k}^{(n)}| \leq \epsilon, k = 1,..., K$, set $d_{k} = d_{k}^{(n+1)}, \forall k$, and compute power control coefficients as $\eta_{k} = \alpha u_{k}/d_{k}, \forall k$. Otherwise, set $n = n + 1$ and go to step 5.
	\end{algorithmic}	
\end{algorithm}
\begin{theorem}\label{theorem:converg_algm3}
	Given $\{\hat{g}_{mk}\}$, $\{\beta_{mk}\}$, $\{\gamma_{mk}\}$, $\rho_u$, $\mathbf{u}$, and $\boldsymbol{\nu}$, if a target-rate power control with target SINR, $S_t$, is feasible, the iterations in step 2 and step 5 of Algorithm \ref{Algorithm 3} converge separately to their unique solutions. The $\{\eta_{k}\}$ given by step 6 of Algorithm \ref{Algorithm 3} achieves this target-rate power control.
\end{theorem}
\begin{proof}
	The proof of Theorem \ref{theorem:converg_algm3} is given in Appendix \ref{Convergence of Algorithm 1 and 3}.
\end{proof}

\subsection{Power Control with Target Rate by RM Approximation 1}
In this section we consider the same settings as in Section \ref{sec:target_exact_SINR}, but use RM Approximation 1. The detailed form is given in {\bf Algorithm} \ref{Algorithm 4}. 
\begin{algorithm}
	\caption{Target Rate Power Control - RM SINR}
	\label{Algorithm 4}
	\begin{algorithmic}[1]
		\item Initialize $\boldsymbol{\nu}$, $\eta_{k}^{(0)}, \forall k$, $\mathbf{D}^{(0)}$ and $\mathbf{T}^{(0)}$ as in step 1 of Algorithm \ref{Algorithm 2}. With a target SINR denoted as $S_{t}$, compute $\alpha = S_{t}M/\rho_{u}$. Set $n = 0$ and choose a tolerance $\epsilon > 0$.
		
		\item Update $\mathbf{T}^{(n+1)}$ using (\ref{eq:T_update}) where $\rho_{u}u_{k'}, \forall k'$ are substituted by $\rho_{u}$ and $\xi_{k'} = S_t,\ \forall k'$.  
		
		\item Stop if $\|\mathbf{T}^{(n+1)} - \mathbf{T}^{(n)}\|_2 < \epsilon$ and set $\mathbf{T} = \mathbf{T}^{(n+1)}$. Otherwise, set $n = n + 1$ and go to step 2.
		
		\item The power control coefficients are computed as $\eta_{k} = \alpha/\text{tr}\left(\nu_{k}\mathbf{\Gamma}_{k}\mathbf{T}\right), \forall k$. If the constraints $0 \leq \eta_{k} \leq 1, \forall k$ are satisfied, Algorithm \ref{Algorithm 4} ends. Otherwise, Initialize $\eta_{k}^{(0)}, \forall k$, $\mathbf{D}^{(0)}$ and $\mathbf{T}^{(0)}$ as in step 1 of Algorithm \ref{Algorithm 2}. Assign the value of each element of vector $\mathbf{u}$ by $u_{g}$ or $u_{p}$ according to the per-device rate under full power case. Set $n = 0, \alpha=\alpha/u_{g}$, and go to step 5.
		
		\item Update $\mathbf{T}^{(n+1)}$ using (\ref{eq:T_update}) where $\xi_{k'} = \frac{S_{t}u_{k'}}{u_g}$.
		
		\item Stop if $\|\mathbf{T}^{(n+1)} - \mathbf{T}^{n}\|_2 \leq \epsilon$, set $\mathbf{T} = \mathbf{T}^{(n+1)}$, and compute power control coefficients by $\eta_{k} = \alpha u_{k}/\text{tr}\left(\nu_{k}\mathbf{\Gamma}_{k}\mathbf{T}\right),\ \forall k$. Otherwise, set $n = n + 1$ and go to step 5.	
	\end{algorithmic}
\end{algorithm}
\begin{theorem}\label{theorem:converg_algm4}
	Given $\{\beta_{mk}\}$, $\{\gamma_{mk}\}$, $\rho_u$, $\mathbf{u}$, and $\boldsymbol{\nu}$, if a target-rate power control with target SINR, $S_t$, is feasible, the iterations given by step 2 and step 5 of Algorithm \ref{Algorithm 4} converge separately to their unique solutions. The $\{\eta_{k}\}$ given by step 6 of Algorithm \ref{Algorithm 4} achieves this target-rate power control.
\end{theorem}
\begin{proof}
	The proof of Theorem \ref{theorem:converg_algm4} is given in Appendix \ref{Convergence of Algorithgm 2 and 4}.
\end{proof}

\subsection{Algorithm Complexity Comparison}
It is noted that in Algorithms \ref{Algorithm 1} and \ref{Algorithm 3} the computation of (\ref{eq:d_updation}) involves an inverse operation of an $M \times M$ non-sparse matrix whose complexity is $\sim\mathcal{O}(M^{3})$. On the other hand, matrix $\mathbf{T}$ is a diagonal matrix, so the computation complexity of (\ref{eq:T_update}) in Algorithms \ref{Algorithm 2} and \ref{Algorithm 4} is  $\sim \mathcal{O}(MK)$. In addition, exact SINR given in (\ref{eq:SINR_k_UL_MMSE}) involves both small-scale and large-scale fading coefficients, frequent updates of the power control coefficients are required. On the contrary, RM Approximation 1 involves only large-scale fading coefficients, power control coefficients can be updated in a much slower rate.

\section{Uplink Simulation Results}\label{Uplink Simulation Results}
\subsection{Setup and Parameters for Numerical Simulations}\label{Setup and Parameters for Numerical Simulations}
We consider networks where $M$ APs and $\bar{K}$ IoT devices are uniformly distributed in a $D \times D \ m^{2}$ square area and $K$ out of $\bar{K}$ devices are active at one moment. The serving area is wrapped around to avoid boundary effects. The large-scale fading coefficients $\beta_{mk}, \forall m, \forall k$ are products of $\text{PL}_{mk}$ and $\text{SF}_{mk}$:
\begin{equation}
\label{eq:large_fading}
\beta_{mk} = \text{PL}_{mk}\text{SF}_{mk},\ \text{with}\ \text{SF}_{mk} = 10^{\frac{\sigma_{\text{sh}}z_{mk}}{10}},
\end{equation}
where $z_{mk} \sim \mathcal{N}(0,1)$. The path loss is generated as in \cite{Ngo_17_Cellfree} where a three-slope model \cite{Tang_01_Pathloss} and the Hata-Cost 231 propagation model \cite{3GPP_ETSI} are used. Shadow fading coefficients are generated based on \cite{Wang08_Joint_Shadow}. The detailed simulation setup and parameters are given in Table \ref{tbl:1}. 
\begin{table}
	\renewcommand{\arraystretch}{1.4}
	\begin{center}
		\caption{Simulation Setup and Parameters.}~\label{tbl:1}
		\fontsize{8}{8}\selectfont
		\begin{tabular}{|>{\centering\arraybackslash}m{6cm}|>{\centering\arraybackslash}m{1.5cm}|>{\centering\arraybackslash}m{0.6cm}|}\hline
			\textbf{Parameter}       &    \textbf{Value}	 \\ \hline
			$f_c$ (Carrier frequency)		 &    1.9 GHz  		     \\ \hline
			BW (Bandwidth) 			     &    20 MHz			 \\ \hline
			Noise figure             &    9 dB               \\ \hline
			$\tau_c$  & 200 \\ \hline
			$P_{u}$ (UL maximum transmit power per data symbol)  &  20 mW             \\ \hline
			$P_{p}$ (UL maximum transmit power per pilot symbol) & 20 mW             \\ \hline
			$P_{d}$ (DL maximum transmit power per AP) & 200 mW \\ \hline
		\end{tabular}
	\end{center}
\end{table}
As performance measures, exact achievable rate for the $k$-th device, $R_{k}^{u, \text{MMSE}}$, and its corresponding throughput, $U_{k}^{u, \text{MMSE}}$, are given by
\begin{equation}
\label{eq:exact_rate}
R_{k}^{u,\text{MMSE}} = \mathbb{E}\left[\log_{2}\left(1 + \text{SINR}_{k}^{u,\text{MMSE}}\right)\right],
\end{equation}
\begin{equation}\label{eq:throughput}
	U_{k}^{u, \text{MMSE}} = \text{BW}\frac{\tau_c - \tau}{2\tau_c}R_{k}^{u, \text{MMSE}},
\end{equation}
where the expectation in (\ref{eq:exact_rate}) is over small-scale fading. On the other hand, the approximated achievable rate of the $k$-th device by RM Approximation 1 is given by $R_{k}^{u,\text{AP1}} = \log_{2}\left(1 + \text{SINR}_{k}^{u,\text{AP1}}\right)$. Throughout all our simulations we assume that pilots $\boldsymbol{\psi}_k$ are generated random $\tau$-tuples with uniform distribution over the surface of a complex unit sphere. 

\subsection{Results and Discussions}
The approximation accuracy of RM Approximation 1 and 2 under full power case in terms of per-device rate is given in Fig. \ref{fig:AP_Exact_Comparison}. 
\begin{figure}
\begin{center}
	\includegraphics [width=0.48\textwidth]{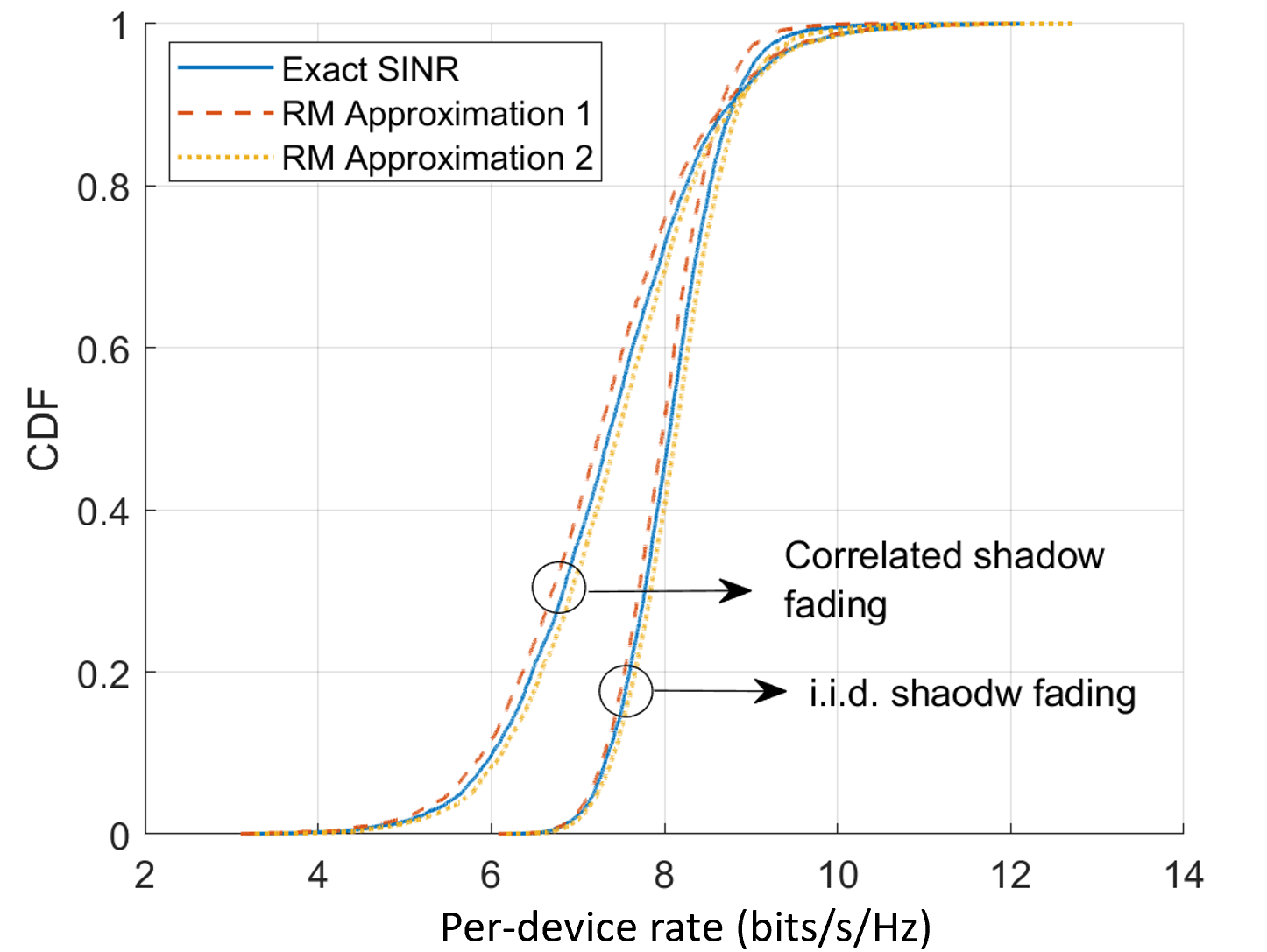}
	\caption{Per-device rate comparison between exact SINR, RM Approximation 1, and RM Approximation 2. Here, $M = 1024$, $K = 256$, $\tau = 256$, and area = 1 $\text{km}^{2}$.}\label{fig:AP_Exact_Comparison}
\end{center}
\end{figure}
For both correlated and i.i.d. shadow fading, the per-device rates obtained using RM Approximation 1 and  2 are quite close to those obtained by exact SINR. This observation verifies Theorem \ref{Theorem 1} in Appendix \ref{Theorems}.

The per-device rate performance comparison under max-min power control based on both exact SINR and RM Approximation 1 is given in Fig. \ref{fig:max_min_power_control}.
\begin{figure}
\begin{center}
\includegraphics [width=0.5\textwidth]{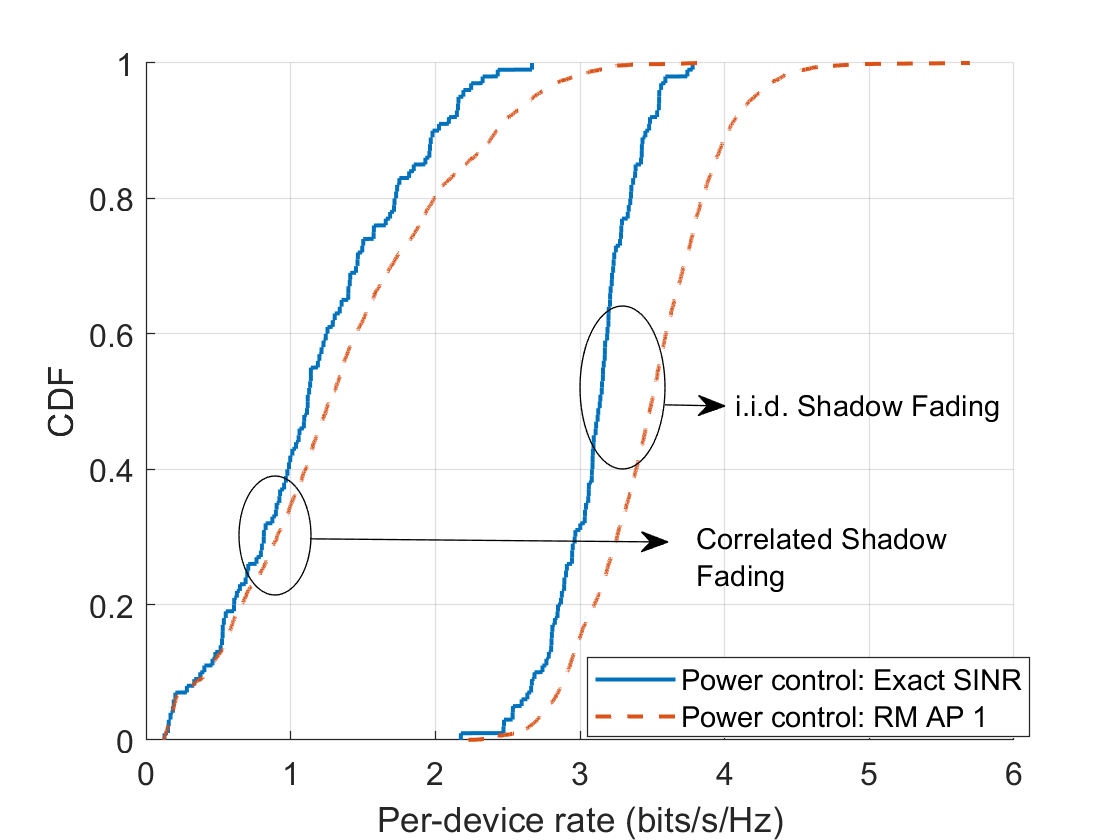} \caption{Per-device rate comparison for max-min power control algorithms based on exact SINR and RM Approximation 1. Here, $M = 160$, $K = 40$, $\tau = 40$, and area = 1 $\text{km}^{2}$.}\label{fig:max_min_power_control}
\end{center}
\end{figure}  
Note that the RM AP 1 curves are obtained as follows. The max-min power control coefficients obtained by Algorithm \ref{Algorithm 2} are substituted into (\ref{eq:exact_rate}) to compute the per-device rates of RM AP 1. Under such a computation, we observe that the per-device rates achieved based on RM AP1 are equivalent to or even better than the performance obtained based on exact SINR. One explanation for this phenomenon is that in order to obtain a uniform service to every device, this uniform rate achieved by max-min power control is limited. Under some realizations of small-scale fading, this rate can be small enough to reduce the expected per-device rate calculated by (\ref{eq:exact_rate}). On the other hand, the power control coefficients used to obtain the curves of RM AP 1 are based on large-scale fading, and as shown in Fig. \ref{fig:max_min_power_control}, better performance can be achieved.

Fig \ref{fig:Opt_subopt_max_min_comparison} shows the per-device throughput comparison of CF IoT systems with different settings. We  consider the cases of i.i.d. and correlated fading cases. 
\begin{figure}
\begin{center}
	\includegraphics
	[width=0.51\textwidth]{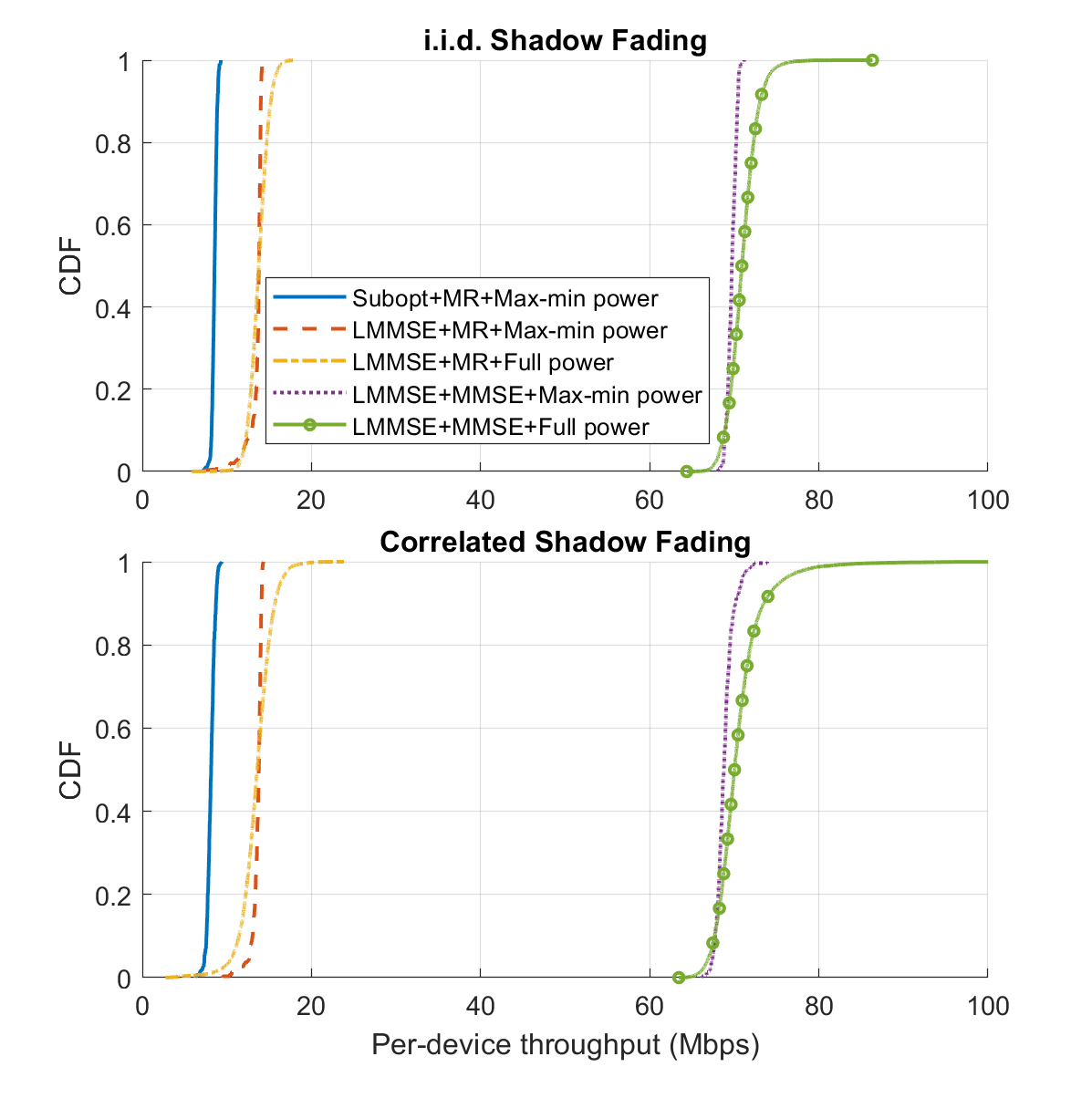} \caption{Performance comparison between optimal and sub-optimal CF IoT systems. Here, $M = 128$, $K = 40$, $\tau = 60$, and area = 0.01 $\text{km}^{2}$.}\label{fig:Opt_subopt_max_min_comparison}
\end{center}
\end{figure} 
In Fig. \ref{fig:Opt_subopt_max_min_comparison}, `Subopt' denotes sub-optimal channel estimation applied in \cite{Ngo_17_Cellfree} and MR is short for maximum-ratio MIMO receiver. It is observed that around 7 times performance improvement is achieved by our system with optimal channel estimation and MMSE MIMO receiver compared with systems with sub-optimal channel estimation and/or MR MIMO receiver. We see, however, that power control does not lead to a significant increase in data rates for both of these cases. However, as it is shown below, power control does lead to a large gain in terms of EE, which is crucial for IoT systems.

The energy efficiencies of the full power transmission case, Algorithm \ref{Algorithm 3} and Algorithm \ref{Algorithm 4} with different target rates are shown in Fig. \ref{fig:EE_Comparison}.
\begin{figure}
\begin{center}
	\includegraphics
	[width=0.52\textwidth]{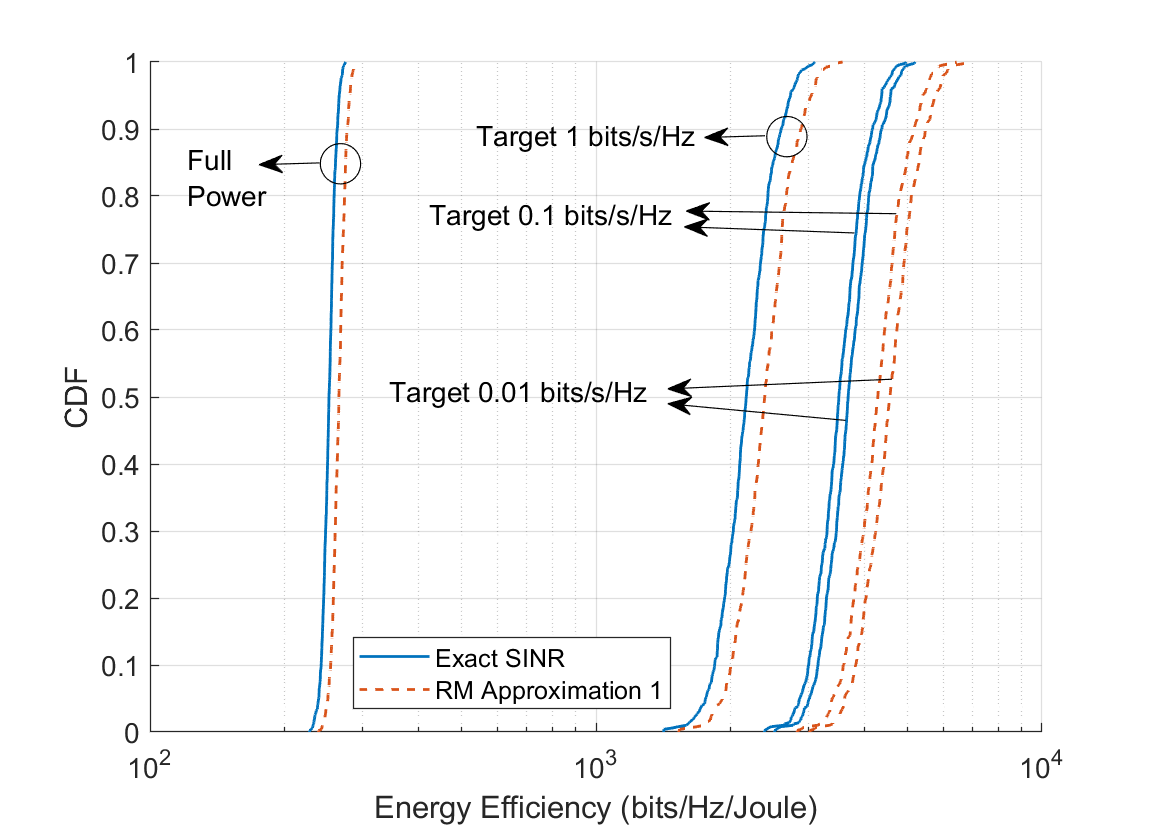} \caption{Energy efficiency comparison among full power case and different target per-device rates under i.i.d. shadow fading. Here, $M = 160$, $K = 40$, $\tau = 40$, and area = 1 $\text{km}^{2}$.}
	\label{fig:EE_Comparison}
\end{center}
\end{figure}
We see that the power control gives large gain over the full power transmission. In particular, for target rates of 0.01 and 0.1 bits/s/Hz we obtain 17-fold and 9-fold improvements, respectively. 
It is also observed that higher EE is obtained by Algorithm \ref{Algorithm 4} (i.e., based on RM Approximation 1) than Algorithm \ref{Algorithm 3} (i.e., based on exact SINR).

\section{Downlink Transmission}\label{Downlink Transmission}
In this section, we consider DL transmission under optimal channel estimation and CB precoding for CF mMIMO IoT systems. 

\subsection{Downlink Data Transmission}
Define $\eta_{mk}, m = 1,2,...,M, k = 1,2,..., K$ as the {\em DL power coefficient} for the data symbol transmitted by the $m$-th AP for the $k$-th device. The transmitted signal from the $m$-th AP is given by
\begin{equation}
	\label{eq:x_m_trans}
	x_m = \sqrt{\rho_{d}}\sum_{k=1}^{K}\sqrt{\eta_{mk}}\hat{g}_{mk}^*s_k.
\end{equation} 
where $s_k$ is the symbol intended for the $k$-th active device and satisfies $\mathbb{E}\{|s_k|^2\} = 1$. Then the received signal at the $k$-th device under CB precoding is given by
\begin{equation}
\label{eq:y_d_k}
	\begin{aligned}
	y_{k}^{d} = 
	&\sqrt{\rho_{d}}\sum_{m=1}^{M}\sqrt{\eta_{mk}}\mathbb{E}[\hat{g}_{mk}^{*}g_{mk}]s_{k} \\ 
	&+\sqrt{\rho_{d}}\sum_{m=1}^{M}\sqrt{\eta_{mk}}\left(\hat{g}_{mk}^{*}g_{mk} - \mathbb{E}[\hat{g}_{mk}^{*}g_{mk}]\right)s_{k} \\
	& + \sqrt{\rho_{d}}\sum_{k'\neq k}\sum_{m=1}^{M}\sqrt{\eta_{mk'}}\hat{g}_{mk'}^{*}g_{mk}s_{k'} + w_{k}^{d}.
	\end{aligned}.
\end{equation}
Based on (\ref{eq:y_d_k}), a closed-form expression for the DL SINR is derived by \cite{Rao_Internet_2019, Rao_cellfree_2018} using the technique in \cite{marzetta2016fundamentals}. This closed-form SINR expression with optimal channel estimation and CB precoding is given in (\ref{eq:SINR_d_k_radom}). 
\begin{figure*}
	\begin{equation}
	\label{eq:SINR_d_k_radom}
	\begin{aligned}
	\text{SINR}_{k}^{\text{IoT}} = \frac{\rho_{d}\left(\sum_{m=1}^{M}\sqrt{\eta_{mk}}\gamma_{mk}\right)^{2}}{\splitfrac{1 + \rho_{d}\sum_{m=1}^{M}\eta_{mk}\gamma_{mk}\beta_{mk} + \rho_{d}\sum_{k'\neq k}\Bigg(\sum_{m=1}^{M}\eta_{mk'}\beta_{mk}||\mathbf{a}_{mk'}||_{2}^{2} + } {\tau\rho_{p}\left(\left|\sum_{m=1}^{M}\sqrt{\eta_{mk'}}\beta_{mk}\boldsymbol{\psi}_{k}^{H}\mathbf{a}_{mk'}\right|^{2} + \sum_{m=1}^{M}\eta_{mk'}\sum_{j=1}^{K}\beta_{mk}\beta_{mj}\left|\boldsymbol{\psi}_{j}^{H}\mathbf{a}_{mk'}\right|^{2}\right)\Bigg)}}
	\end{aligned}.
	\end{equation}
\hrulefill
\end{figure*}
When the devices are assigned orthonormal pilots, it can be verified that $\text{SINR}_{k}^{\text{IoT}}$ is converted to $\text{SINR}_{k}^{\text{orth}}$,
\begin{equation}
\label{eq:SINR_d_k_orth}
\text{SINR}_{k}^{\text{orth}} = \frac{\rho_{d}\left(\sum_{m=1}^{M}\sqrt{\eta_{mk}\gamma_{mk}}\right)^{2}}{1 + \rho_{d}\sum_{k'=1}^{K}\sum_{m=1}^{M}\eta_{mk'}\gamma_{mk'}\beta_{mk}},	
\end{equation}
which coincides with the SINR derived in \cite{Ngo_17_Cellfree}.

\section{Downlink Power Control}\label{Downlink Power Control}
In the DL transmission, we consider the following per-AP power constraint:
\begin{equation}
	\label{eq:pow_constr_DL}
	\mathbb{E}\{|x_m|^2\} \leq \rho_{d},\, m = 1,2,...,M.
\end{equation}
With the model defined in (\ref{eq:ch_coeffi}), (\ref{eq:pow_constr_DL}) can be rewritten as
\begin{equation}
	\label{eq:pow_constr_DL_2}
	\sum_{k=1}^{K}\eta_{mk}\gamma_{mk} \leq 1, m = 1,2,...,M.
\end{equation}

\subsection{Optimal Power Control}
According to the power constraint given in (\ref{eq:pow_constr_DL_2}), DL max-min power control based on random pilots under an IoT system can be formulated as the problem shown below:
\begin{equation}
\begin{aligned}
	\label{eq:max_min_dl_IoT}
	&\max_{\boldsymbol{\eta}}\min_{k = 1,..., K}\text{SINR}_{k}^{\text{IoT}}\\
	&\text{s.t.} \sum_{k' = 1}^{K}\eta_{mk'}\gamma_{mk'} \leq 1,\ m = 1,2,..., M\\
	&\;\;\;\;\;\eta_{mk'} \geq 0,\ m = 1,2,..., M, k' = 1,2,..., K.
\end{aligned}
\end{equation}
The optimization problem (\ref{eq:max_min_dl_IoT}) is quasi-concave \cite{Ngo_17_Cellfree}. It can be solved by performing a bisection search and solving a convex feasibility problem in each step \cite{boyd2004convex}. However, as the number of APs and IoT devices increases, the bisection search becomes too complex, and a more simple power control algorithm is required. For orthonormal pilots, the max-min power control is also a quasi-concave problem and is shown below:
\begin{equation}
\begin{aligned}
\label{eq:max_min_dl_orth}
&\max_{\boldsymbol{\eta}}\min_{k}\text{SINR}_{k}^{\text{orth}}(\boldsymbol{\eta}) = \frac{\rho_{d}\left(\sum_{m=1}^{M}\sqrt{\eta_{mk}}\gamma_{mk}\right)^{2}}{1 + \rho_{d}\sum_{m = 1}^{M}\beta_{mk}\sum_{k'=1}^{K}\eta_{mk'}\gamma_{mk'}}\\
&\qquad\qquad\text{s.t.} \sum_{k' = 1}^{K}\eta_{mk'}\gamma_{mk'} \leq 1, m = 1,..., M\\
& \qquad\qquad\quad\;\eta_{mk'} \geq 0, m = 1,..., M, k' = 1,..., K.	
\end{aligned}	
\end{equation}

\subsection{Power Control using Neural Network}
As we noticed above, the complexity of finding the optimal solution of (\ref{eq:max_min_dl_IoT}) or (\ref{eq:max_min_dl_orth}) is too high for any practical applications. In this Section we suggest to use Neural Networks for finding low complexity suboptimal power control. 

We first let $p_m= \sum_{k=1}^K \eta_{mk} \gamma_{mk}$ be the normalized transmit power of the $m$-th AP and let $p_m^{\text{opt}}$ be the optimal value of $p_m$ with respect to the optimization problem (\ref{eq:max_min_dl_orth}). It is noticeable that if we can find $p_{m}^{\text{opt}},\ \forall m$, the problem in (\ref{eq:max_min_dl_orth}) becomes equivalent to
\begin{equation}
	\begin{aligned}
		\label{eq:max_min_dl_orth_convex}
		&\max_{\boldsymbol{\eta}}\min_{k}\text{SINR}_{k}^{\text{orth}}(\boldsymbol{\eta}) = \frac{\rho_{d}\left(\sum_{m=1}^{M}\sqrt{\eta_{mk}}\gamma_{mk}\right)^{2}}{1 + \rho_{d}\sum_{m=1}^{M}p_{m}^{\text{opt}}\beta_{mk}}\\
		&\qquad\qquad\text{s.t.} \sum_{k'=1}^{K} \eta_{mk'}\gamma_{mk'} = p_{m}^{\text{opt}}\\
		& \qquad\qquad\quad\;\eta_{mk'} \geq 0, m = 1,..., M, k' = 1,..., K,
	\end{aligned}
\end{equation}
which is a convex problem \cite{Nayebi2017Precoding} and has significant smaller complexity compared with the quasi-concave problem (\ref{eq:max_min_dl_orth}).
 
To convert problem (\ref{eq:max_min_dl_orth}) to a convex problem, $p_{m}^{\text{opt}}, \forall m$, need to be found. It is observed in \cite{Nayebi2017Precoding} that an exponential relation often approximately holds between $\beta_{m}^{\text{max}}$ and $p_{m}^{\text{opt}}$ where $\beta_{m}^{\text{max}}$ is defined as the largest large-scale fading coefficient between the $m$-th AP and its serving devices, i.e., $\beta_{m}^{\text{max}} = \max_{k=1,.., K}\beta_{mk}$. An exponential regression can then be implemented to predict $p_{m}^{\text{opt}}, \forall m$. We denote the outputs of the exponential regression as $p_{m}(\beta_{m}^{\text{max}}), \forall m$, and these outputs can be used in (\ref{eq:max_min_dl_orth_convex}). In addition, we find that as the length of random pilots goes to infinity and assume that random pilots $\boldsymbol{\psi}_k$ are used in IoT systems, the value of $\text{SINR}_{k}^{\text{IoT}}$ will approach to the value of  $\text{SINR}_{k}^{\text{orth}}$, i.e.,
\begin{equation}
	\label{eq:IoT_to_orth}
	\lim_{\tau\rightarrow \infty} \text{SINR}_{k}^{\text{IoT}} \longrightarrow \text{SINR}_{k}^{\text{orth}}.
\end{equation}
It is also noted that even with finite, but reasonably large $\tau$, $\text{SINR}_{k}^{\text{IoT}} \approx \text{SINR}_{k}^{\text{orth}}$. Based on these observations, low complexity power control for DL IoT systems with random pilots can be implemented as follows. First, $p_{m}^{\text{opt}},\, \forall m$, are approximated using exponential regression based on $\beta_{m}^{\text{max}},\, \forall m$, as in \cite{Nayebi2017Precoding}. The outputs, $p_{m}(\beta_{m}^{\text{max}}),\, \forall m$, are then substituted into (\ref{eq:max_min_dl_orth_convex}), and solving the obtained convex optimization problem, we find the power coefficients $\eta_{mk}$.

However, the performance achieved using $p_{m}(\beta_{m}^{\text{max}})$ is not close enough to the optimal performance achieved by $p_{m}^{\text{opt}}$. What is even more important is that the generality of this method is limited. In particular, if we found a function $p_m(\beta_m^{\text{max}})$ that matches well $p_m^{\text{opt}}$ for one network, typically this function is not accurate for another network. Moreover, for some networks, no exponential relationship can be found between $p_{m}^{\text{opt}}$ and $\beta_{m}^{\text{max}}$, e.g., this is the case for high density networks. 

To overcome the above problems, we propose to use a simple fully connected NN to approximate $p_{m}^{\text{opt}}, \forall m$. As shown in Fig.~\ref{fig:NN}, the structure of the NN we used includes three hidden fully connected layers and each layer has four neurons.
\begin{figure*}
\begin{center}
	\includegraphics [width=0.75\textwidth]{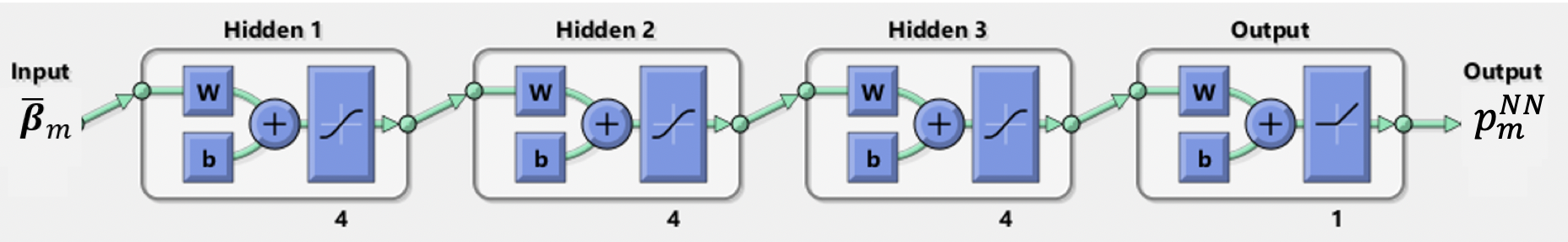}
	\caption{Neural network for predicting	 $p_{m}^{\text{NN}}, \forall m$.}\label{fig:NN}
\end{center}
\hrulefill
\end{figure*}
Tangent-sigmoid activation function is used for the three hidden layers and rectified linear unit (ReLU) activation function is used in the output layer. The input vector of the NN $\bar{\boldsymbol{\beta}}_m = [\bar{\beta}_{m1}, \bar{\beta}_{m2},...,\bar{\beta}_{m\widehat{K}}]^T \in \mathbb{R}^{\widehat{K} \times 1}$ consists of $\widehat{K}$ largest large-scale fading coefficients from the $m$-th AP to its nearby serving devices. For example, $\bar{\beta}_{m1} = \max_{k=1,.., K}\beta_{mk}$, $\bar{\beta}_{m2}$ is the second largest large-scale fading coefficient and so on. By $p_m^\text{NN}$ we denote the approximation of $p_m^{\text{opt}}$ predicted by the NN for the $m$-th AP.

\subsection{Scalable Power Control with High Energy Efficiency}\label{sec:Scalable_pow_control}
In real-life applications, we expect that the areas covered by CF networks will be most likely very large. Such large areas will contain so large number of APs and IoT devices that even the computation complexity of solving the convex problem (\ref{eq:max_min_dl_orth_convex}) would be too high. Thus a scalable power control algorithm whose complexity grows linearly with $M$ and $K$ is required.

Here we propose a sub-optimal scalable power control algorithm that achieves high EE at the same time. We first define the {\em density of a network} as
\begin{equation}
\label{eq:density_network}
\text{Density} = \frac{\text{Number of APs}}{\text{Serving Area}}.
\end{equation}
We also assume that the ratio $M/K$ is fixed. As mentioned above, for very large networks even the convex problem (\ref{eq:max_min_dl_orth_convex}) becomes too heavy. For such networks we replace the max-min power optimization with the uniform power control, in which the power coefficients are given by
\begin{equation}
	\label{eq:uniform_power_control}
	\eta_{m} = \frac{p_{m}^{\text{NN}}}{\sum_{k=1}^{K}\gamma_{mk}},\ \eta_{m} = \eta_{mk},\, \forall k,
\end{equation}
where $p_m^\text{NN}$ is also obtained using the NN structure given in Fig.~\ref{fig:NN}. It is important to note that the computation complexity of predicting every $p_m^{\text{NN}}$ by the NN is not only very low due to the simple structure of the NN, but also keeps constant as $M$ and $K$ increase. 
This means that we obtain a scalable network since the amount of computations conducted by each AP does not depend on the number of APs in the network. At the same time, it will be shown later that (\ref{eq:uniform_power_control}) provides much higher EE compared with the full power case.

\subsection{Neural Network Training}
In the process of training, $\widehat{K}$ is set as four and Levenberg-Marquardt algorithm \cite{levenberg1944method}, \cite{marquardt1963algorithm} is used to train the NN. We use $10^4$ training samples $p_m^\text{opt}$ obtained by finding the optimal solutions of (\ref{eq:max_min_dl_orth}). After the training we find the NN weights $\boldsymbol{\theta}_{j}^{n_j}, \forall n_j, \forall j$ with $\boldsymbol{\theta}_{j}^{n_j}$ being the weight vector of the $n_j$-th neuron in $j$-th layer and these weights are applied for online prediction of $p_m^\text{NN}$. This approach is used for producing results presented in Fig.~\ref{fig:NN_Any_Area}.

For training the NN for scalable power control, the same training parameters and algorithm mentioned above are adopted. The idea is to train NN for small areas, which have a relatively small number of APs and therefore afford solving (\ref{eq:max_min_dl_orth}) and finding $p_m^{\text{opt}}$ needed for training. The area is wrapped around, see details in \cite{Ngo_17_Cellfree}, in order to mimic an infinite size network. Next, we use {\em the same NN} for large areas that have the same density and are not wrapped around. This means that we use the same weights $\boldsymbol{\theta}_{j}^{n_j}, \forall n_j, \forall j$, and each AP, say AP $m$, uses input vector $\bar{\boldsymbol{\beta}}_m = [\bar{\beta}_{m1}, \bar{\beta}_{m2},...,\bar{\beta}_{m\widehat{K}}]^T$ composed by $\widehat{K}$ largest coefficients $\beta_{mk}$.

Note that the NN is trained offline before online application. Thus, the online complexity of the trained NN only comes from the inference stage (i.e., the prediction of $p_m^{\text{opt}}$) which is very low due to the simple NN structure applied.

\subsection{Simulation Results and Discussions}  
The same simulation setup and parameters given in Section \ref{Setup and Parameters for Numerical Simulations} are used here. 

{\bf Experiment 1: } In this experiment, we would like to train a NN so that it can be applied to any given serving area with a fixed number of APs and devices. In our experiment we used $M=128$ and $K=4$. Since the number of APs and devices is relatively small, for given powers $p_m^\text{opt}$, we can find power coefficients by solving the problem (\ref{eq:max_min_dl_orth_convex}). We train our NN for several areas and next we use the obtained NN for an area that was not used for training. The results of this approach are shown in Fig. \ref{fig:NN_Any_Area}.
\begin{figure}
	\begin{center}
		\includegraphics [width=0.5\textwidth]{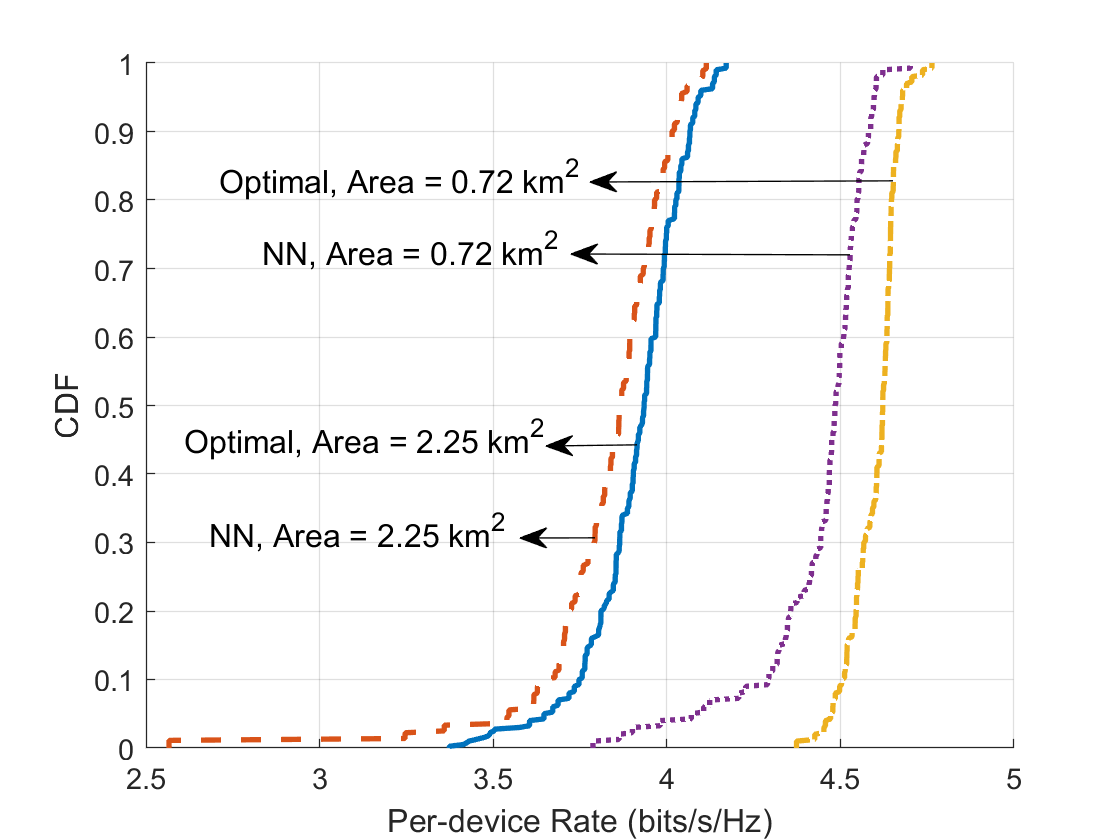}
		\caption{Performance comparison between NN based and 	optimal power control algorithms of DL CF IoT systems. Here, $M = 128$, $K = 4$.}\label{fig:NN_Any_Area}
	\end{center}
\end{figure} 
The NN used in Fig. \ref{fig:NN_Any_Area} is trained by squares with areas: $0.016, 0.063, 0.25, 0.56, 1, 4\ \text{km}^2$ and then the obtained NN is used for squares of sizes $0.72$ and $2.25\ \text{km}^2$. We see that this NN provides the performance which is close to the performance obtained with the optimal power control.  
%The performance gaps are within $5\%$ in terms of $5\%$ outage rate for both $0.72\ \text{km}^{2}$ and $2.25\ \text{km}^{2}$ cases.

{\bf Experiment 2:} In this experiment, we fixed the density of the serving area as $\frac{M = 64}{0.03\,km^2}$ and train the NN with $M = 64$ and $K = 16$. Then we use the trained NN for finding $p_m^\text{NN}$ for networks covering large areas with the same density and the same ratio $M/K$, see Fig.~\ref{fig:NN_fix_density_large}.

First, we compare our NN approach with other approaches for a small area and a small $M$ and $K$, in Fig.~\ref{fig:NN_fix_density}. We use this small area, $M$, and $K$ since this allows us to find the solution of (\ref{eq:max_min_dl_orth}) and (\ref{eq:uniform_power_control}), which are necessary for the comparisons in Fig.~\ref{fig:NN_fix_density}.
\begin{figure}
	\begin{center}
		\includegraphics [width = 0.5\textwidth]{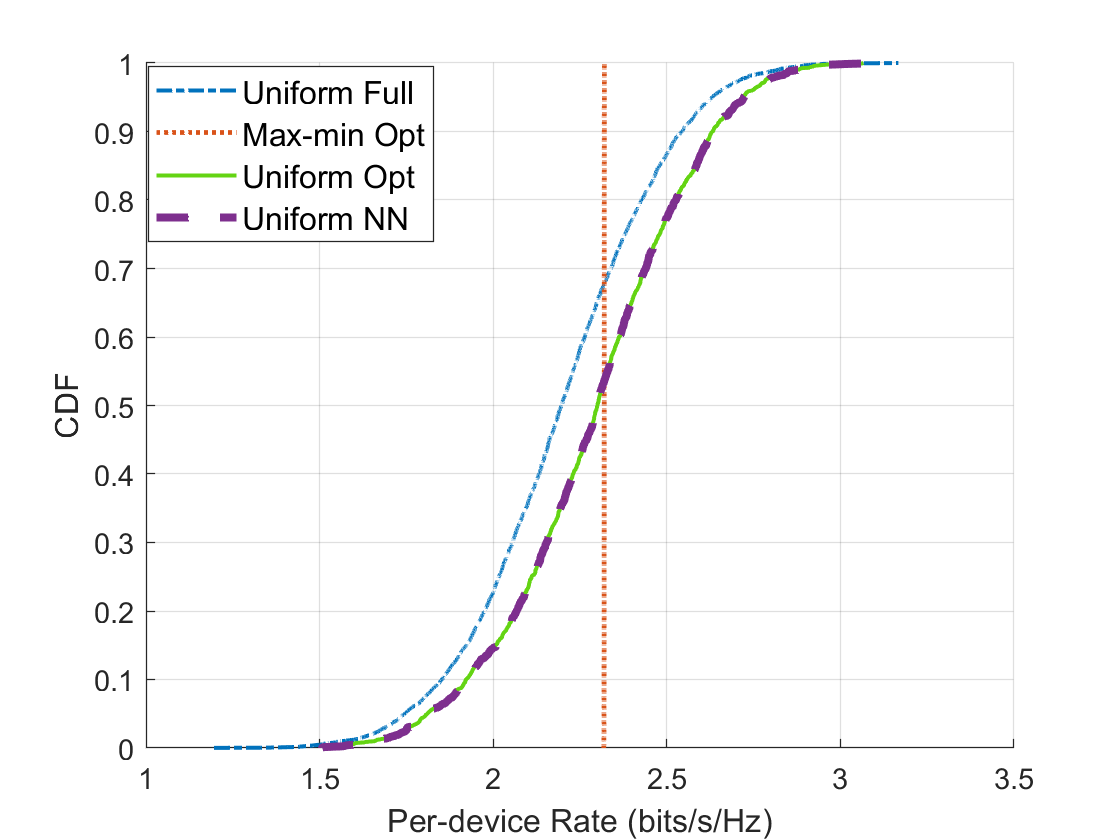}
		\caption{Spectral efficiency comparison among full power, uniform power control, and optimal power control cases for DL CF IoT systems. Here, $M = 64$, $K = 16$, $\text{Area} = 0.03\ \text{km}^{2}$.}\label{fig:NN_fix_density}
	\end{center}
	%\vspace{-2.5mm}
\end{figure}
\begin{figure}
	%\vspace{-5mm}
	\begin{center}
		\includegraphics [width=0.5\textwidth]{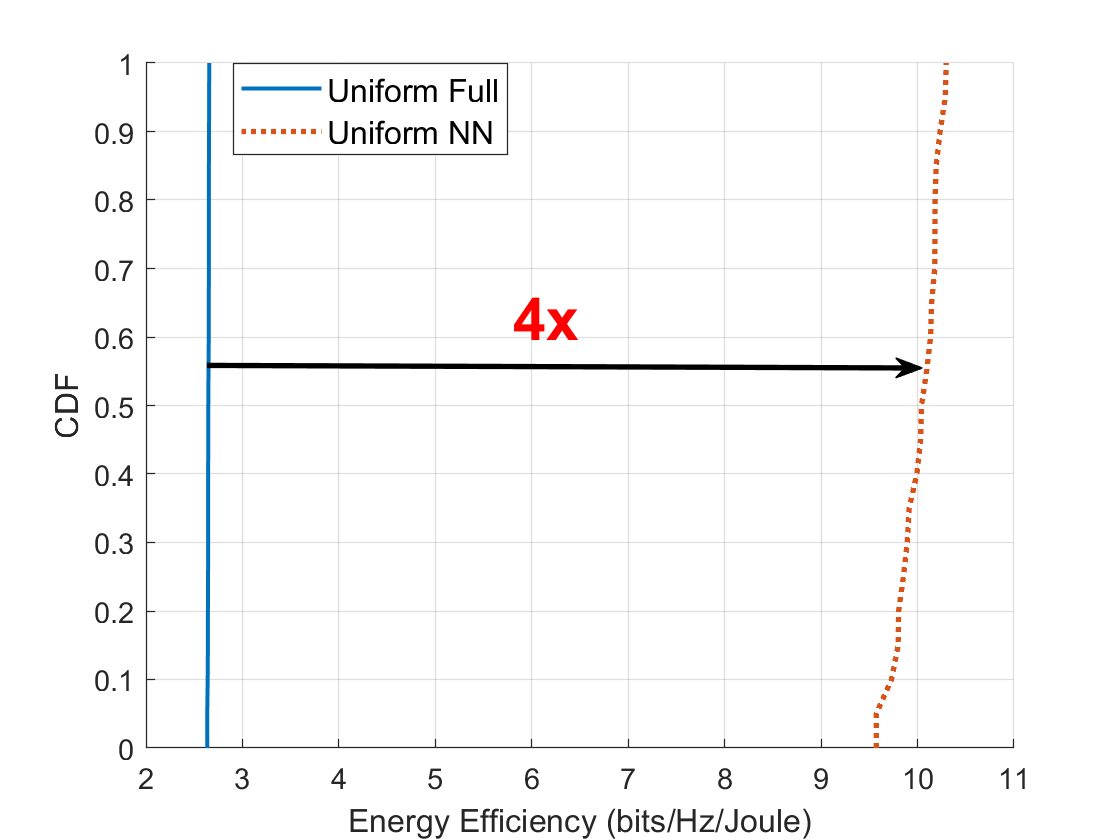}
		\caption{Energy efficiency comparison between full power and uniform power control for DL CF IoT systems. Here, $M = 4096$, $K = 1024$, $\text{Area} = 2\ \text{km}^{2}$.}\label{fig:NN_fix_density_large}
	\end{center}
\end{figure}
Note that the vertical appearance of ``Max-min Opt" is due to equalization of the data rates of all $K$ devices.
We present rates when we use (\ref{eq:uniform_power_control}) with maximal $p_m$ (Uniform Full), optimal $p_m^{\text{opt}}$ (Uniform Opt), and NN produced $p_m^{\text{NN}}$ (Uniform NN). It is observed from Fig.~\ref{fig:NN_fix_density} that under uniform power control, the SE performance is almost the same using either $p_{m}^{\text{NN}}, \forall m$, or $p_{m}^{\text{opt}}, \forall m$, which suggests accurate predictions of the proposed NN. Although the performance gap between uniform power control and full power case is not big under current relatively high area density case, we found that the gap would increase in low area density case. More importantly, the NN power control allows us to significantly improve EE compared with the full power transmission as it is shown in Fig.~\ref{fig:NN_fix_density_large}. Note that EE is crucially important for IoT networks.

Similar to UL transmission, we define DL EE as $$E_d \triangleq \frac{\sum_{k=1}^{K}R_{k}^{d}}{\sum_{m=1}^{M}P_m}$$ where $P_m$ can be the maximum transmit power of the $m$-th AP (i.e, $P_d$) or $p_m^{\text{NN}}P_d$. Fig.~\ref{fig:NN_fix_density_large} shows the EE comparison between uniform power control and full power case for a network with the same density $\frac{M = 64}{0.03\,km^2}$ and $M=4096$ and $K=1024$. We observe that our NN power control leads to 4-fold improvement of EE.

\section{Conclusion}\label{Conclusion}
In this work, we proposed IoT systems supported by CF mMIMO with optimal components - LMMSE channel estimation and MMSE MIMO receiver. We derived two random matrix approximations for device's UL SINR and used one of them for efficient and low complexity power control algorithms that give large EE gains, which is very important for low powered IoT devices. 
For DL transmission a NN aided max-min power control algorithm is proposed. Comparing with the optimal max-min power control algorithm, it has significantly lower complexity and achieves comparable performance.
We further proposed a scalable NN algorithm for transmit power control. This algorithm, though sub-optimal, 
incorporates a simple fully connected NN and obtains power coefficients with very low complexity. The scalability of this algorithm is also important for IoT systems since it works for a very large service area (i.e., covering a large number of IoT devices). Multifold gain in EE is obtained by this algorithm compared with the full power transmission approach.

\appendix
\subsection{Some Useful Lemmas}
\begin{lemma} \label{Lemma 1}[Matrix Inversion Lemma] \cite[(2.2)]{silverstein1995empirical}
	Let $\mathbf{U}$ be an $M \times M$ invertible matrix and $\mathbf{x} \in \mathbb{C}^{M\times 1}$, $c\in \mathbb{C}$ for which $\mathbf{U} + c\mathbf{xx}^{H}$ is invertible. Then
	\begin{equation}
	\mathbf{x}^{H}\left(\mathbf{U} + c\mathbf{xx}^{H}\right)^{-1} = \frac{\mathbf{x}^{H}\mathbf{U}^{-1}}{1 + c\mathbf{x}^{H}\mathbf{U}^{-1}\mathbf{x}}.
	\end{equation}
\end{lemma}

\begin{lemma}\label{Lemma 2}\cite[Lemma 4 and 5]{Wagner_12_MISO}
	Let $\mathbf{A} \in \mathbb{C}^{M\times M}$ and $\mathbf{x},\mathbf{y} \sim \mathcal{CN} (0, \frac{1}{M}\mathbf{I}_{M}).$ Assume that $\mathbf{A}$ has uniformly bounded spectral norm (with respect to $M$) and that $\mathbf{x}$ and $\mathbf{y}$ are mutually independent and independent of $\mathbf{A}$. Then,
	\begin{align}
	\mathbf{x}^{H}\mathbf{A}\mathbf{x} - \frac{1}{M}\text{tr}\left(\mathbf{A}\right)\xrightarrow[M \rightarrow \infty]{\text{a.s.}} 0,
	\end{align}
	\begin{align}
	\label{eq:x_A_y}
	\mathbf{x}^{H}\mathbf{A}\mathbf{y} \xrightarrow[M \rightarrow \infty]{\text{a.s.}} 0.
	\end{align}
\end{lemma}

\begin{lemma}\label{Lemma 3}\cite[Lemma 6]{Wagner_12_MISO}
Let $\mathbf{A}_{1}, \mathbf{A}_{2},...,$ with $\mathbf{A}_{M} \in \mathbb{C}^{M\times M}$, be deterministic with uniformly bounded spectral norm and $\mathbf{B}_{1}, \mathbf{B}_{2},...,$ with $\mathbf{B}_{M} \in \mathbb{C}^{M\times M}$, be random Hermitian, with eigenvalues $\lambda_{1}^{\mathbf{B}_{M}}\leq \lambda_{2}^{\mathbf{B}_{M}} \cdots \leq \lambda_{M}^{\mathbf{B}_{M}}$ such that, with probability 1, there exist $\epsilon > 0$ for which $\lambda_{1}^{\mathbf{B}_{M}} > \epsilon$ for all large $M$. Then for $v \in \mathbb{C}^{M\times 1}$
\begin{equation}
\frac{1}{M}\text{tr}\left(\mathbf{A}_{M}\mathbf{B}_{M}^{-1}\right) - \frac{1}{M}\text{tr}\left(\mathbf{A}_{M}\left(\mathbf{B}_{M} + \mathbf{vv}^{H}\right)^{-1}\right) \xrightarrow[M \rightarrow \infty]{\text{a.s.}} 0
\end{equation} 
almost surely, where $\mathbf{B}_{M}^{-1}$ and $\left(\mathbf{B}_{M} + \mathbf{vv}^{H}\right)^{-1}$ exist with probability 1.	
\end{lemma}

\subsection{Two Theorems \label{Theorems}}
\begin{theorem}\label{Theorem 1}
Let $m = 1,..., M$ and let $\hat{g}_{mk}$ and $\hat{g}_{ml}$ be two channel coefficients estimated using random pilots and LMMSE channel estimation as given in Section \ref{System Model}. Let $\tau$ be the length of the random pilots. Then, for a fixed $m$,
\begin{equation}
\label{eq:theorem_1}
\text{Cov}[\hat{g}_{mk},\hat{g}_{ml}] \xrightarrow[\tau \rightarrow \infty]{\text{a.s.}} 0 \ \text{for}\ k,\,l = 1,...,K\ \text{and}\ k\neq l.
\end{equation}
\end{theorem}

\begin{proof}
\begin{equation}
	\label{eq:cov_g}
	\begin{aligned}
		&\text{cov}[\hat{g}_{mk},\hat{g}_{ml}] \\
		&= \mathbb{E}\bigg[\Big(\mathbf{a}_{mk}^{H}\Big(\sqrt{\tau\rho_{p}}\sum_{i=1}^{K}g_{mi}\boldsymbol{\psi}_{i} + \mathbf{w}_{m}\Big)\Big)^{H}\cdot\\
		&\qquad\quad\quad\mathbf{a}_{ml}^{H}\Big(\sqrt{\tau\rho_{p}}\sum_{j=1}^{K}g_{mj}\boldsymbol{\psi}_{j} + \mathbf{w}_{m}\Big)\bigg]\\
		&= \underbrace{\sum_{i=1}^{K}\sum_{j=1}^{K}\tau\rho_{p}\mathbb{E}\left[g_{mi}^{*}g_{mj}\right]\boldsymbol{\psi}_{i}^{H}\mathbf{a}_{mk}\mathbf{a}_{ml}^{H}\boldsymbol{\psi}_{j}}_{T_{1}} +\\ &\underbrace{\mathbb{E}\left[\mathbf{w}_{m}^{H}\mathbf{a}_{mk}\mathbf{a}_{ml}^{H}\mathbf{w}_{m}\right]}_{T_{2}} + \underbrace{(\text{cross terms of $\mathbb{E}$ of $g_{mi}^*$ and $\mathbf{w}_{m}$ })}_{T_{3}}.
	\end{aligned}
\end{equation}
Due to the independence between $g_{mi}$ and each component of $\mathbf{w}_{m}$, $T_{3} = 0$. On the other hand, $g_{mi}$ and $g_{mj}$ are also independent, we get:
\begin{equation}
\begin{aligned}
	\label{eq:T_1}
	T_{1} &= \sum_{i=1}^{K}\tau\rho_{p}\mathbb{E}\left[|g_{mi}|^2\right]\boldsymbol{\psi}_{i}^{H}\mathbf{a}_{mk}\mathbf{a}_{ml}^{H}\boldsymbol{\psi}_{i}\\
	&=\sum_{i=1}^{K}\tau\rho_{p}\beta_{mi}\boldsymbol{\psi}_{i}^{H}\mathbf{a}_{mk}\mathbf{a}_{ml}^{H}\boldsymbol{\psi}_{i}.
\end{aligned}
\end{equation}
Define
\begin{equation}
\label{eq:Z_m}
\mathbf{Z}_{m} = \tau\rho_{p}\mathbf{\Psi}\mathbf{B}_{m}\mathbf{\Psi}^{H} + \mathbf{I}_{\tau} = \tau\rho_{p}\sum_{k=1}^{K}\beta_{mk}\boldsymbol{\psi}_{k}\boldsymbol{\psi}_{k}^{H} + \mathbf{I}_{\tau}.
\end{equation}
Then $\mathbf{a}_{mk} = \sqrt{\tau\rho_{p}}\beta_{mk}\mathbf{Z}_{m}^{-1}\boldsymbol{\psi}_{k}$. Let $\mathbf{Z}_{mi} = \mathbf{Z}_{m} - \tau\rho_{p}\beta_{mi}\boldsymbol{\psi}_{i}\boldsymbol{\psi}_{i}^{H}$, and apply (\ref{eq:Z_m}) and lemma 1 to $\boldsymbol{\psi}_{i}^{H}\mathbf{a}_{mk}$, we obtain:
\begin{equation}
\label{eq:psi_a}
	\boldsymbol{\psi}_{i}^{H}\mathbf{a}_{mk} = \frac{\sqrt{\tau\rho_{p}}\beta_{mk}\boldsymbol{\psi}_{i}^{H}\mathbf{Z}_{mi}^{-1}\boldsymbol{\psi}_{k}}{1 + \tau\rho_{p}\beta_{mi}\boldsymbol{\psi}_{i}^{H}\mathbf{Z}_{mi}^{-1}\boldsymbol{\psi_{i}}}.
\end{equation}
Let $\mathbf{Z}_{mik} = \mathbf{Z}_{mi} - \tau\rho_{p}\beta_{mk}\boldsymbol{\psi}_{k}\boldsymbol{\psi}_{k}^{H}$ and apply it together with lemma 1 to $\boldsymbol{\psi}_{i}^{H}\mathbf{Z}_{mi}^{-1}$ in the numerator of (\ref{eq:psi_a}), we obtain:
\begin{equation}
\begin{aligned}
	\label{eq:psi_a_2}
	&\boldsymbol{\psi}_{i}^{H}\mathbf{a}_{mk} =\\ &\frac{\sqrt{\tau\rho_{p}}\beta_{mk}\boldsymbol{\psi}_{i}^{H}\mathbf{Z}_{mik}^{-1}\boldsymbol{\psi}_{k}}{\left(1 + \tau\rho_{p}\beta_{mi}\boldsymbol{\psi}_{i}^{H}\mathbf{Z}_{mi}^{-1}\boldsymbol{\psi}_{i}\right)\left(1 + \tau\rho_{p}\beta_{mk}\boldsymbol{\psi}_{k}^{H}\mathbf{Z}_{mik}^{-1}\boldsymbol{\psi}_{k}\right)}.
\end{aligned}
\end{equation}
In (\ref{eq:psi_a_2}), $\mathbf{Z}_{mik}$ does not depend on $\boldsymbol{\psi}_{i}$ and $\boldsymbol{\psi}_{k}$, and $\boldsymbol{\psi}_{i}$ and $\boldsymbol{\psi}_{k}$ are Gaussian and independent. Hence, by applying (\ref{eq:x_A_y}), we can get:
\begin{equation}
	\boldsymbol{\psi}_{i}^{H}\mathbf{a}_{mk}\xrightarrow[\tau \rightarrow \infty]{\text{a.s.}} 0\;\text{for}\; i\neq k.
\end{equation}
In the same way, $\mathbf{a}_{ml}^{H}\boldsymbol{\psi}_{i} \xrightarrow[\tau \rightarrow \infty]{\text{a.s.}} 0$ for $i \neq l$. Since $l\neq k$, we obtain $T_{1} \xrightarrow[\tau \rightarrow \infty]{\text{a.s.}} 0$ (if $K$ is finite).

Now we consider $T_{2}$:
\begin{equation}
\begin{aligned}
\label{eq:T_2}
T_{2} =\;& \mathbb{E}\left[\text{tr}\left(\mathbf{w}_{m}\mathbf{w}_{m}^{H}\mathbf{a}_{mk}\mathbf{a}_{ml}^{H}\right)\right] = \text{tr}\left(\mathbb{E}\left[\mathbf{w}_{m}\mathbf{w}_{m}^{H}\right]\mathbf{a}_{mk}\mathbf{a}_{ml}^{H}\right) \\=\;& \text{tr}\left(\mathbf{I}_{\tau}\mathbf{a}_{mk}\mathbf{a}_{ml}^{H}\right)
= \mathbf{a}_{ml}^{H}\mathbf{a}_{mk}\\ 
=\;&\tau\rho_{p}\beta_{mk}\beta_{ml}\boldsymbol{\psi}_{l}^{H}\mathbf{Z}_{m}^{-1}\mathbf{Z}_{m}^{-1}\boldsymbol{\psi}_{k}.
\end{aligned}
\end{equation}
Let $\mathbf{Z}_{ml}^{-1} = \mathbf{Z}_{m} - \tau\rho_{p}\beta_{ml}\boldsymbol{\psi}_{l}\boldsymbol{\psi}_{l}^{H}$, and apply it together with lemma 1 to $\boldsymbol{\psi}_{l}^{H}\mathbf{Z}_{m}^{-1}$, we obtain:
\begin{equation}
\label{eq:psi_Z}
\boldsymbol{\psi}_{l}^{H}\mathbf{Z}_{m}^{-1} = \frac{\boldsymbol{\psi}_{l}^{H}\mathbf{Z}_{ml}^{-1}}{1 + \tau\rho_{p}\beta_{ml}\boldsymbol{\psi}_{l}^{H}\mathbf{Z}_{ml}^{-1}\boldsymbol{\psi}_{l}}.
\end{equation}
Since $\boldsymbol{\psi}_{l}$ is Gaussian and does not depend on $\mathbf{Z}_{ml}$, we obtain:
\begin{equation}
	\label{eq:psi_Z_as}
	\boldsymbol{\psi}_{l}^{H}\mathbf{Z}_{m}^{-1} \xrightarrow[\tau \rightarrow \infty]{\text{a.s.}} 0.
\end{equation}  
In the same way $\mathbf{Z}_{m}^{-1}\boldsymbol{\psi}_{k} \xrightarrow[\tau \rightarrow \infty]{\text{a.s.}} 0$, so $T_{2} \xrightarrow[\tau \rightarrow \infty]{\text{a.s.}} 0$. Thus we prove that $\text{cov}[\hat{g}_{mk},\hat{g}_{ml}] \xrightarrow[\tau \rightarrow \infty]{\text{a.s.}} 0$ (if $K$ is finite).  
\end{proof}

\begin{theorem}\label{Theorem 2} \cite{Hoydis_13_UL/DL}, \cite{Wagner_12_MISO}
Let $\mathbf{S}_{M} \in \mathbb{C}^{M\times M}$ be Hermitian non-negative definite and let $\mathbf{X} \in \mathbb{C}^{M \times K}$ be random with independent column vectors $\mathbf{x}_{k} \sim \mathcal{CN} (0, \frac{1}{M}\mathbf{R}_{k})$. Define $\mathbf{Q}_{M} \in \mathbb{C}^{M\times M}$ deterministic and assume that $\mathbf{Q}_{M}$ and $\mathbf{R}_{k},\, k = 1,..., K$ have uniformly bounded spectral norms (with respect to M). Then, for any $z > 0$, 
\begin{equation}
\begin{aligned}
	&\frac{1}{M}\text{tr}\left(\mathbf{Q}_{M}\left(\mathbf{X}\mathbf{X}^{H} + \mathbf{S}_{M} + z\mathbf{I}_{M}\right)^{-1}\right) - \frac{1}{M}\text{tr}\left(\mathbf{Q}_{M}\mathbf{T}(z)\right)\\
	&\xrightarrow[M \rightarrow \infty]{\text{a.s.}} 0,
\end{aligned}
\end{equation}
where $\mathbf{T}(z) \in \mathbb{C}^{M\times M}$ is given by
\begin{equation}
\mathbf{T}(z) = \left(\frac{1}{M}\sum_{k=1}^{K}\frac{\mathbf{R}_{k}}{1 + e_{k}(z)} + \mathbf{S}_{M} + z\mathbf{I}_{M}\right)^{-1}.
\end{equation}
Here, $e_{k}(z) = \lim_{t\rightarrow \infty} e_{k}^{(t)}(z)$ and $e_{k}^{(t)}(z)$ is obtained by
\begin{equation}
e_{k}^{(t)} = \frac{1}{M}\text{tr}\Bigg(\mathbf{R}_{k}\bigg(\frac{1}{M}\sum_{k'=1}^{K}\frac{\mathbf{R}_{k'}}{1 + e_{k'}^{(t-1)}(z)} + \mathbf{S}_{M} + z\mathbf{I}_{M}\bigg)^{-1}\Bigg).
\end{equation} 
where $t = 1, 2,...$ and $e_{k}^{(0)}(z) = 1/z$ for $k = 1, 2,..., K$. 
\end{theorem}
\begin{proof}
	The proof of Theorem \ref{Theorem 2} is given in \cite{Wagner_12_MISO}. 
\end{proof}

\subsection{Derivations for RM Approximations}\label{Proof for RM Approximations}
We first derive RM Approximation 2. Define $\boldsymbol{\Omega} = \rho_{u}\sum_{k'=1}^{K}\eta_{k'}\hat{\mathbf{g}}_{k'}\hat{\mathbf{g}}_{k'}^{H} + \mathbf{D}$, $\boldsymbol{\Omega}_{k} = \rho_{u}\sum_{k'\neq k}^{K}\eta_{k'}\hat{\mathbf{g}}_{k'}\hat{\mathbf{g}}_{k'}^{H} + \mathbf{D}$, where $\mathbf{D}$ is defined in (\ref{eq:def_D}).  Using (\ref{eq:est_err_g}) we get that  $\beta_{mk} - \gamma_{mk} \geq 0$. Hence the sum in $\mathbf{D}$ is a non-negative definite matrix and hence $\frac{\mathbf{D}}{M}$ can be considered as $\mathbf{S}_M + z \mathbf{I}_M$ used in Theorem \ref{Theorem 2}.  Using Lemma \ref{Lemma 1}, we can represent $\rho_{u}\eta_{k}\hat{\mathbf{g}}_{k}^{H}\boldsymbol{\Omega}^{-1}\hat{\mathbf{g}}_{k}$ in the numerator and denominator of (\ref{eq:SINR_k_UL_MMSE}) as  
\begin{equation}
\label{eq:lemma1}
\rho_{u}\eta_{k}\hat{\mathbf{g}}_{k}^{H}\boldsymbol{\Omega}^{-1}\hat{\mathbf{g}}_{k} = \frac{\rho_{u}\eta_{k}\hat{\mathbf{g}}_{k}^{H}\boldsymbol{\Omega}_{k}^{-1}\hat{\mathbf{g}}_{k}}{1 + \rho_{u}\eta_{k}\hat{\mathbf{g}}_{k}^{H}\boldsymbol{\Omega}_{k}^{-1}\hat{\mathbf{g}}_{k}}.
\end{equation}
Using Lemma \ref{Lemma 2}, we get
\begin{equation}
\label{eq:lemma2}
\rho_{u}\eta_{k}\hat{\mathbf{g}}_{k}^{H}\boldsymbol{\Omega}_{k}^{-1}\hat{\mathbf{g}}_{k} - \rho_{u}\eta_{k}\text{tr}\left(\mathbf{\Gamma}_{k}\boldsymbol{\Omega}_{k}^{-1}\right) \xrightarrow[M \rightarrow \infty]{\text{a.s.}} 0.
\end{equation}
It is noted that, in CF mMIMO IoT systems, $\text{Cov}[\hat{g}_{mk}, \hat{g}_{ml}]$ is not strictly 0 when random pilots are applied. However, based on Theorem \ref{Theorem 1}, we can directly use Theorem \ref{Theorem 2} to derive RM approximations. Applying Theorem \ref{Theorem 2} to (\ref{eq:lemma2}), we can get
\begin{equation}
\rho_{u}\eta_{k}\text{tr}\left(\mathbf{\Gamma}_{k}\boldsymbol{\Omega}_{k}^{-1}\right) - \frac{\rho_{u}\eta_{k}}{M}\text{tr}\left(\mathbf{\Gamma}_{k}\mathbf{T}_{k}\right) \xrightarrow[M \rightarrow \infty]{\text{a.s.}} 0.
\end{equation}    
where $\mathbf{T}_{k}$ is defined in (\ref{eq:T_k}). Then $\rho_{u}\eta_{k}\hat{\mathbf{g}}_{k}^{H}\boldsymbol{\Omega}_{k}^{-1}\hat{\mathbf{g}}_{k}$ is substituted by $\frac{\rho_{u}\eta_{k}}{M}\text{tr}\left(\mathbf{\Gamma}_{k}\mathbf{T}_{k}\right)$ in (\ref{eq:lemma1}), and (\ref{eq:lemma1}) is further substituted into (\ref{eq:SINR_k_UL_MMSE}) to obtain RM Approximation 2 in (\ref{eq:RMT_Approx2}). 

RM Approximation 1 is obtained by using Lemma \ref{Lemma 3}. First we note that $\boldsymbol{\Omega}_k$ is positive definite. Indeed, let ${\bf v} =  [v_1, v_2,..., v_M]^T \in \mathbb{C}^{M \times 1}$ be any non-zero vector, then 
\begin{align*}
	\mathbf{v}^H\boldsymbol{\Omega}_k \mathbf{v} = &\rho_u \sum_{k'\neq k}^{K}\eta_{k'}|\mathbf{v}^H\hat{\mathbf{g}}_{k'}|^2\\
	& + \rho_u\sum_{k'=1}^{K}\eta_{k'}\sum_{m=1}^{M}|v_m|^2 (\beta_{mk'} - \gamma_{mk'}) + \|{\bf v}\|_2^2.
\end{align*}
Since $\rho_u > 0, 0 \leq \eta_{k'} \leq 1,\ \forall k' , |\mathbf{v}^H\hat{\mathbf{g}}_{k'}| \geq 0, \beta_{mk} - \gamma_{mk'} \geq 0$ and $\|\mathbf{v}\|_2^2 > 0$, we obtain $\mathbf{v}^H\boldsymbol{\Omega}_k \mathbf{v} > 0$. Thus $\boldsymbol{\Omega}_k$ is positive definite. 
Now applying Lemma \ref{Lemma 3} to $\rho_{u}\eta_{k}\text{tr}\left(\mathbf{\Gamma}_{k}\mathbf{\Omega}_{k}^{-1}\right)$, we get
\begin{equation}
\rho_{u}\eta_{k}\text{tr}\left(\mathbf{\Gamma}_{k}\mathbf{\Omega}_{k}^{-1}\right) - \rho_{u}\eta_{k}\text{tr}\left(\mathbf{\Gamma}_{k}\mathbf{\Omega}^{-1}\right) \xrightarrow[M \rightarrow \infty]{\text{a.s.}} 0.
\end{equation}
Using Theorem \ref{Theorem 2} again, we get
\begin{equation}
\rho_{u}\eta_{k}\text{tr}\left(\mathbf{\Gamma}_{k}\mathbf{\Omega}^{-1}\right) - \frac{\rho_{u}\eta_{k}}{M}\text{tr}\left(\mathbf{\Gamma}_{k}\mathbf{T}\right) \xrightarrow[M \rightarrow \infty]{\text{a.s.}} 0. 
\end{equation}
where $\mathbf{T}$ is defined in (\ref{eq:T_Approx1}). Then $\rho_{u}\eta_{k}\hat{\mathbf{g}}_{k}^{H}\boldsymbol{\Omega}_{k}^{-1}\hat{\mathbf{g}}_{k}$ is substituted by $\frac{\rho_{u}\eta_{k}}{M}\text{tr}\left(\mathbf{\Gamma}_{k}\mathbf{T}\right)$ in (\ref{eq:lemma1}), and (\ref{eq:lemma1}) is further substituted into (\ref{eq:SINR_k_UL_MMSE}) to obtain RM Approximation 1 in (\ref{eq:RMT_Approx1}).

\subsection{Proof of Theorem \ref{converg_algm1} and \ref{theorem:converg_algm3}}\label{Convergence of Algorithm 1 and 3}
Theorem \ref{converg_algm1} and \ref{theorem:converg_algm3} are proved using the framework of standard interference functions \cite{yates1995framework}. The process is as follows. Define $\mathbf{d} = [d_{1}, d_{2},..., d_{K}]^{T} \in \mathbb{R}^{K}$ and $\mathbf{f}(\mathbf{d}) = [f_{1}(\mathbf{d}),f_{2}(\mathbf{d}),..., f_{K}(\mathbf{d})]^{T} \in \mathbb{R}^{K}$, where
\begin{equation}
\label{eq:SIF_d}
	f_{k}(\mathbf{d}) = \rho'_{u,k}\hat{\mathbf{g}}_{k}^{H}\left(\alpha\sum_{k'=1}^{K}\frac{\rho'_{u,k'}u_{k'}}{d_{k'}}\mathbf{J}_{k'} + \mathbf{I}_{M}\right)^{-1}\hat{\mathbf{g}}_{k}.
\end{equation}
It is noted that (\ref{eq:SIF_d}) is equivalent to (\ref{eq:d_updation}). According to \cite[Theorem 1 and 2]{yates1995framework}, the iteration $\mathbf{d}^{(n+1)} = \mathbf{f}(\mathbf{d}^{(n)})$ will converge to a unique point $\mathbf{d}^{\star} = [d_1^{\star}, d_2^{\star},..., d_K^{\star}]^T \in \mathbb{R}^K$ for any non-negative initial point $d_{k}^{(0)}, \forall k$ if and only if $\mathbf{f}(\mathbf{d})$ is a feasible standard interference function. $\mathbf{f}(\mathbf{d})$ is a standard interference function if for all $\mathbf{d} \geq 0$ the following properties are satisfied.
 
1. Positivity: $\mathbf{f}(\mathbf{d}) > 0$,
 
2. Monotonicity: if $\mathbf{d} > \mathbf{d}'$, then $\mathbf{f}(\mathbf{d}) > \mathbf{f}(\mathbf{d}')$,

3. Scalability: For all $\zeta > 1, \zeta\mathbf{f}(\mathbf{d}) > \mathbf{f}(\zeta\bar{\mathbf{d}})$.

Here, the vector inequality $\mathbf{d} > \mathbf{d}'$ is a strict inequality in all components.

\subsubsection{Proof of Theorem \ref{converg_algm1}} In Algorithm \ref{Algorithm 1}, $\alpha = \min_{k'}(d_{k'}/u_{k'}),\, k' = 1,..., K$. Although $\alpha$ might be different for different iterations, it can be regarded as a positive constant for each iteration. We will first prove that when $\alpha$ is a positive constant, $\mathbf{f}(\mathbf{d})$ is a standard interference function.

Define $\mathbf{F} = \alpha\sum_{k'=1}^{K}(\rho'_{u,k'}u_{k'}/d_{k'})\mathbf{J}_{k'} + \mathbf{I}_{M}$.
Based on (\ref{eq:est_err_g}) and (\ref{eq:J_matrix}), it is straightforward to verify that $\mathbf{F}$ is Hermitian positive definite, then $\mathbf{F}^{-1}$ is also Hermitian positive definite, so $\mathbf{f}(\bar{\mathbf{d}}) > 0$, and the positivity property is proved. 

It is noted that the monotonicity property is equivalent to $\partial \mathbf{f(d)}/\partial \mathbf{d} > \mathbf{0}$ for all possible $\mathbf{d}$, where $\mathbf{0}$ is a matrix with each component being 0 and the matrix inequality $\partial \mathbf{f(d)}/\partial \mathbf{d} > \mathbf{0}$ is a strict inequality in all components. Denote the $(k, k'')$-th component of $\partial \mathbf{f(d)}/\partial \mathbf{d}$ as $\partial f_{k}(\mathbf{d})/\partial d_{k''}$, where $k = 1,...,K,$ and $k''= 1,...,K$. Then
\begin{align*}
\frac{\partial f_{k}(\mathbf{d})}{\partial d_{k''}}& = \frac{\partial \rho'_{u,k}\mathbf{\hat{g}}_{k}^{H}\mathbf{F}^{-1}\mathbf{\hat{g}}_{k}}{\partial d_{k''}} = \rho'_{u,k}\mathbf{\hat{g}}_{k}^{H}\frac{\partial\mathbf{F}^{-1}}{\partial d_{k''}}\mathbf{\hat{g}}_{k}\\
&=-\rho'_{u,k}\mathbf{\hat{g}}_{k}^{H}\mathbf{F}^{-1}\frac{\partial\mathbf{F}}{\partial d_{k''}}\mathbf{F}^{-1}\mathbf{\hat{g}}_{k}\\
=&\frac{\alpha(\rho'_{u,k})^{2}u_{k''}}{d_{k''}^2}\mathbf{\hat{g}}_{k}^{H}\mathbf{F}^{-1}\left(\mathbf{\hat{g}}_{k''}\mathbf{\hat{g}}_{k''}^{H} + \mathbf{B}_{k''} - \boldsymbol{\Gamma}_{k''}\right)\mathbf{F}^{-1}\mathbf{\hat{g}}_{k}.
\end{align*} 
Since $\mathbf{\hat{g}}_{k''}\mathbf{\hat{g}}_{k''}^{H} + \mathbf{B}_{k''} - \boldsymbol{\Gamma}_{k''}$ is Hermitian positive definite and $d_{k''}^{2} \geq 0$ for all $k''$, we obtain $\partial f_{k}(\mathbf{d})/\partial d_{k''} > 0, \forall k$ and $\forall k''$. Thus $\mathbf{f(d)}$ satisfies the monotonicity property.

The proof for $\mathbf{f(d)}$ satisfying the scalability property is as follows. Define $\mathbf{\bar{J}} = \alpha\sum_{k'=1}^{K}(\rho'_{u,k'}u_{k'}/d_{k'})\mathbf{J}_{k'}$, then $f_{k}(\zeta\mathbf{d})$ can be written as
\begin{align*}
f_k(\zeta \mathbf{d}) =& \rho'_{u,k}\mathbf{\hat{g}}_{k}^{H}\left(\frac{1}{\zeta}\mathbf{\bar{J}} +  \mathbf{I}_{M}\right)^{-1}\mathbf{\hat{g}}_{k}\\
=& \rho'_{u,k}\mathbf{\hat{g}}_{k}^{H}\left(\mathbf{Q}\frac{\boldsymbol{\Lambda}}{\zeta} \mathbf{Q}^{-1} +  \mathbf{Q}\mathbf{Q}^{-1}\right)^{-1}\mathbf{\hat{g}}_{k}\\
=& \rho'_{u,k}\mathbf{\hat{g}}_{k}^{H}\mathbf{Q}\left(\zeta\boldsymbol{\Lambda}^{-1} + \mathbf{I}_{M}\right) \mathbf{Q}^{H}\mathbf{\hat{g}}_{k},\\
\end{align*} 
where $\mathbf{Q}\boldsymbol{\Lambda} \mathbf{Q}^{-1}$ is the eigenvalue decomposition of $\mathbf{\bar{J}}$. Following the same approach, $\zeta f_k(\mathbf{d})$ can be written as
\begin{align*}
\zeta f_k(\mathbf{d}) = \rho'_{u,k}\mathbf{\hat{g}}_{k}^{H}\mathbf{Q}\zeta\left(\boldsymbol{\Lambda}^{-1} + \mathbf{I}_{M}\right) \mathbf{Q}^{H}\mathbf{\hat{g}}_{k}.
\end{align*}
Since $\mathbf{\bar{J}}$ is Hermitian positive definite, each component of $\diag\{\boldsymbol{\Lambda}^{-1} + \mathbf{I}_{M}\}$ is positive. Based on this fact, we obtain $\diag\{\zeta\left(\boldsymbol{\Lambda}^{-1} + \mathbf{I}_{M}\right)\} > \diag\{\zeta\boldsymbol{\Lambda}^{-1} + \mathbf{I}_{M}\} \ \text{for}\ \zeta > 1$, where the vector inequality is a strict inequality in all components. Since both $\zeta\left(\boldsymbol{\Lambda}^{-1} + \mathbf{I}_{M}\right)$ and $\zeta\boldsymbol{\Lambda}^{-1} + \mathbf{I}_{M}$ are positive definite, we obtain $\zeta f_k(\mathbf{d}) > f_k(\zeta\mathbf{d}), \forall k\ \text{for}\ \zeta > 1$. Thus it is proved that $\mathbf{f(d)}$ satisfies the scalability property and we also prove that $\mathbf{f}(\mathbf{d})$ is a standard interference function when $\alpha$ is a positive constant. 

The convergence of Algorithm \ref{Algorithm 1} is proved as follows. For a given network realization, denote $R^{\star}$ as the max-min per-device rate which can be achieved by all served active devices. If $R^{\star}$ is known in advance, we can teat $R^{\star}$ as a target rate and the value of $\alpha$ corresponding to it is denoted as $\alpha^{\star}$. It is noted that when $\alpha = \alpha^{\star}$ in (\ref{eq:SIF_d}), the iteration $\mathbf{d}^{(n+1)} = \mathbf{f}(\mathbf{d}^{(n)})$ will converge to a unique solution of UL power control coefficients, $\boldsymbol{\eta}^{\star}$ where $\max_k(\eta_k^{\star}) = 1,\ k = 1, 2, ..., K$. 

On the other hand, two operations are implemented by step 2 of Algorithm \ref{Algorithm 1} for every iteration. We define $\alpha^{(n)} = \min_{k'}(d_{k'}^{(n - 1)}/u_{k'})$. Replacing $\alpha$ by $\alpha^{(n)}$ in (\ref{eq:SIF_d}) we obtain $\mathbf{f}^{(n)}(\mathbf{d})$. The first operation is the computation of $\alpha^{(n)}$ which is actually an update of power control coefficients which are expressed as 
\begin{equation}
	\label{eq:eta_update}
	\eta_k^{(n)} = \frac{\min_{k'}\left(d_{k'}^{(n-1)}/u_{k'}\right)}{d_{k}^{(n-1)}/u_k},\ \forall k,\ k' = 1, 2,..., K.
\end{equation}
The second operation is using $\alpha^{(n)}$ (i.e., $\boldsymbol{\eta}^{(n)} = [\eta_1^{(n)}, \eta_2^{(n)},...,\eta_K^{(n)}]$) to compute $\mathbf{d}^{(n)}$ by (\ref{eq:d_updation}). Note that $\mathbf{f}^{(n)}(\mathbf{d})$ is a standard interference function as proved earlier, it has a unique solution which is denoted as $(\boldsymbol{\eta}^{(n)})^{\star}$. 
It is also noted that if $\alpha^{(n)}$ is not equal to $\alpha^{\star}$, $\max_k((\eta^{(n)}_k)^{\star})$ is also not equal to 1. 
The first operation makes $\max_k(\eta_{k}^{(n)}) = 1$, which leads to $ \alpha^{(n)} \xrightarrow[n \rightarrow \infty]{\text{a.s.}} \alpha^{\star}$ (i.e., $ (\boldsymbol{\eta}^{(n)})^{\star} \xrightarrow[n \rightarrow \infty]{\text{a.s.}} \boldsymbol{\eta}^{\star}$).  
On the other hand,  the second operation  makes $\boldsymbol{\eta}^{(n)}$ converges to $(\boldsymbol{\eta}^{(n)})^{\star}$ since $\mathbf{f}^{(n)}(\mathbf{d})$ is a standard interference function. Thus we conclude that by the iteration of Algorithm \ref{Algorithm 1}, $\boldsymbol{\eta}^{(n)}$ converges to $\boldsymbol{\eta}^{\star}$.

\subsubsection{Proof of Theorem \ref{theorem:converg_algm3}}\label{Convergence of Algorithm 3}
In Algorithm \ref{Algorithm 3}, $\alpha = S_t/(1+S_t)$. 
Replacing $\rho'_{u,k'}u_{k'},\  \forall k'$ in (\ref{eq:SIF_d}) by $\rho'_{u,k'},\ \forall k'$, we obtain $\bar{\mathbf{f}}(\mathbf{d})$ which is equivalent to the formula used in Step 2 of Algorithm \ref{Algorithm 3}. Since $\alpha$ is a positive constant, using the same approach for $\mathbf{f}(\mathbf{d})$, $\bar{\mathbf{f}}(\mathbf{d})$ can also prove to be a standard interference function. Thus Theorem \ref{theorem:converg_algm3} is proved.

\subsection{Proof of Theorem \ref{theorem: converg_algm2} and \ref{theorem:converg_algm4}}\label{Convergence of Algorithgm 2 and 4}
Theorem \ref{theorem: converg_algm2} and \ref{theorem:converg_algm4} can also be proved using the framework of standard interference functions \cite{yates1995framework}.
The process is as follows. 
Define $\mathbf{l} = [l_{1}, l_{2},..., l_{M}]^{T} \in \mathbb{R}^{M}$ and $\mathbf{q}(\mathbf{l}) = [q_{1}(\mathbf{l}),q_{2}(\mathbf{l}),..., q_{M}(\mathbf{l})]^{T} \in \mathbb{R}^{M}$, where
\begin{equation}
\begin{aligned}
\label{eq:SIF_RM1}
\small
&q_{m}(\mathbf{l}) = \\ &\Bigg[\Bigg(\frac{\alpha}{M}\sum_{k'=1}^{K}\frac{\rho_{u}u_{k'}}{\text{tr}\left(\mathbf{\Gamma}_{k'}\diag\{\mathbf{l}\}\right)}\Big(\mathbf{B}_{k'}- \frac{\xi_{k'}\mathbf{\Gamma}_{k'}}{1 + \xi_{k'}}\Big) + \frac{\mathbf{I}_{M}}{M}\Bigg)^{-1}\Bigg]_{mm}
\end{aligned}
\end{equation}
It is noted that (\ref{eq:SIF_RM1}) is equivalent to (\ref{eq:T_update}). 
Here $\xi_{k'}$ can either depend on $\alpha$ as in Algorithm \ref{Algorithm 2} or be a positive constant as in Algorithm \ref{Algorithm 4}, and $[\mathbf{A}]_{mm}$ means the $m$-th diagonal element of matrix $\mathbf{A}$. 
\subsubsection{Proof of Theorem \ref{theorem: converg_algm2}} In Algorithm \ref{Algorithm 2}, $\alpha = \min_{k'}\left(\text{tr}\left(\nu_{k'}\mathbf{\Gamma}_{k'}\diag\{\mathbf{l}\}\right)/u_{k'}\right), \,k' = 1,..., K$. 
The prove process is similar to that for Algorithm \ref{Algorithm 1}. 
We will first prove that when $\alpha$ is a positive constant, $\mathbf{q}(\mathbf{l})$ is a standard interference function.

According to (\ref{eq:est_err_g}), $\beta_{mk}-\gamma_{mk} \geq 0, \forall m, \forall k$, so $\left(\mathbf{B}_{k'}- \frac{e_{k'}\mathbf{\Gamma}_{k'}}{1 + e_{k'}}\right) > 0$. 
Since $\xi_{k'},\ \forall k'$ in Algorithm \ref{Algorithm 2} are also non-negative, we conclude that $\mathbf{q}(\mathbf{l}) > 0$ for all $\mathbf{l} \geq 0$ and the positivity property is proved.

%On the other hand, according to the definition and positivity of $\mathbf{q}(\mathbf{l})$, it is easy to find that $\mathbf{q}(\zeta\mathbf{l}) = \mathbf{q}(\mathbf{l}) < \zeta\mathbf{q}(\mathbf{l})$ for all $\zeta > 1$. Thus the scalability property is also proved.

Since $\mathbf{B}_{k'}$, $\boldsymbol{\Gamma}_{k'}$, and $\mathbf{I}_M$ are diagonal matrices, $q_m(\mathbf{l})$ can be easily rewritten as
\begin{equation}
\label{eq:q_m_L}
q_m(\mathbf{l}) = \frac{M}{\alpha\sum_{k'=1}^{K}\frac{\rho_{u}u_{k'}}{\text{tr}\left(\boldsymbol{\Gamma}_{k'}\diag\{\mathbf{l}\}\right)}\left(\beta_{mk'}-\frac{\xi_{k'}\gamma_{mk'}}{1+\xi_{k'}}\right) + 1}.
\end{equation}
Thus it is easy to verify that for all $\zeta > 1$, $\zeta q_m(\mathbf{l}) > q_m(\zeta \mathbf{l}), \forall m$ and the scalabiltity property is proved.

As for the monotonicity, if $\mathbf{l} > \mathbf{l}'$, then $\frac{\rho_{u}u_{k'}}{\text{tr}\left(\boldsymbol{\Gamma}_{k'}\diag\{\mathbf{l}\}\right)} < \frac{\rho_{u}u_{k'}}{\text{tr}\left(\boldsymbol{\Gamma}_{k'}\diag\{\mathbf{l'}\}\right)},\ \forall k'$. 
Since we assume $\alpha$ is a positive constant, so does $\xi_{k'},\ \forall k'$. Then $q_m(\mathbf{l}) > q_m(\mathbf{l}'), \forall m$ (i.e., $\mathbf{q(l)} > \mathbf{q(l')}$) and the monotonicity property is proved.
Thus we prove that $\mathbf{q(l)}$ is a standard interference function when $\alpha$ is a positive constant.

Using the same approach as that for Algorithm \ref{Algorithm 1}, we can obtain that by the iteration in Step 2 of Algorithm $\ref{Algorithm 2}$, $ \alpha^{(n)} \xrightarrow[n \rightarrow \infty]{\text{a.s.}} \alpha^{\star}$ (i.e., $ (\boldsymbol{\eta}^{(n)})^{\star} \xrightarrow[n \rightarrow \infty]{\text{a.s.}} \boldsymbol{\eta}^{\star}$), and $\boldsymbol{\eta}^{(n)}$ converges to $(\boldsymbol{\eta}^{(n)})^{\star}$. Here $\alpha^{(n)} = \min_{k'}\left(\text{tr}\left(\nu_{k'}\mathbf{\Gamma}_{k'}\mathbf{T}^{(n-1)}\right)/u_{k'}\right), \,k' = 1,..., K$ and $\eta_k^{(n)}$ is defined as:
\begin{align*}
	\eta_k^{(n)} =  \frac{\min_{k'}\left(\text{tr}\left(\nu_{k'}\mathbf{\Gamma}_{k'}\mathbf{T}^{(n-1)}\right)/u_{k'}\right)}{\text{tr}\left(\nu_{k}\mathbf{\Gamma}_{k}\mathbf{T}^{(n-1)}\right)/u_{k}},\ \forall k,\ k' = 1,..., K.
\end{align*}
Thus we conclude that by the iteration of Algorithm \ref{Algorithm 2}, $\boldsymbol{\eta}^{(n)}$ converges to $\boldsymbol{\eta}^{\star}$.

\subsubsection{Proof of Theorem \ref{theorem:converg_algm4}} In Algorithm \ref{Algorithm 4}, $\alpha = S_{t}M/\rho_{u}$ or $S_{t}M/(\rho_{u} u_g)$. 
Replacing $\rho_{u}u_{k'},\  \forall k'$ in (\ref{eq:SIF_RM1}) by $\rho_{u}$, we obtain $\bar{\mathbf{q}}(\mathbf{l})$ which is equivalent to the formula used in step 2 of Algorithm \ref{Algorithm 4}. Since $\alpha$ and $\xi_{k'},\ \forall k'$ are positive constants in Algorithm \ref{Algorithm 4}, using the same approach for $\mathbf{q}(\mathbf{l})$, $\bar{\mathbf{q}}(\mathbf{l})$ can also prove to be a standard interference function. Thus Theorem \ref{theorem:converg_algm4} is proved.

\bibliographystyle{IEEEtran}
\bibliography{references}

\begin{thebibliography}{10}
\providecommand{\url}[1]{#1}
\csname url@samestyle\endcsname
\providecommand{\newblock}{\relax}
\providecommand{\bibinfo}[2]{#2}
\providecommand{\BIBentrySTDinterwordspacing}{\spaceskip=0pt\relax}
\providecommand{\BIBentryALTinterwordstretchfactor}{4}
\providecommand{\BIBentryALTinterwordspacing}{\spaceskip=\fontdimen2\font plus
\BIBentryALTinterwordstretchfactor\fontdimen3\font minus
  \fontdimen4\font\relax}
\providecommand{\BIBforeignlanguage}[2]{{%
\expandafter\ifx\csname l@#1\endcsname\relax
\typeout{** WARNING: IEEEtran.bst: No hyphenation pattern has been}%
\typeout{** loaded for the language `#1'. Using the pattern for}%
\typeout{** the default language instead.}%
\else
\language=\csname l@#1\endcsname
\fi
#2}}
\providecommand{\BIBdecl}{\relax}
\BIBdecl

\bibitem{Chettri2020Compre}
L.~{Chettri} and R.~{Bera}, ``A comprehensive survey on internet of things
  (iot) toward 5g wireless systems,'' \emph{IEEE Internet of Things Journal},
  vol.~7, no.~1, pp. 16--32, 2020.

\bibitem{Palattella_IoT_2016}
M.~R. {Palattella}, M.~{Dohler}, A.~{Grieco}, G.~{Rizzo}, J.~{Torsner},
  T.~{Engel}, and L.~{Ladid}, ``{Internet of Things in the 5G Era: Enablers,
  Architecture, and Business Models},'' \emph{IEEE Journal on Selected Areas in
  Communications}, vol.~34, no.~3, pp. 510--527, 2016.

\bibitem{Akpakwu_Sruvey_2018}
G.~A. {Akpakwu}, B.~J. {Silva}, G.~P. {Hancke}, and A.~M. {Abu-Mahfouz}, ``{A
  Survey on 5G Networks for the Internet of Things: Communication Technologies
  and Challenges},'' \emph{IEEE Access}, vol.~6, pp. 3619--3647, 2018.

\bibitem{Wang2019User}
L.~{Wang}, Y.~{Ai}, N.~{Liu}, and A.~{Fei}, ``User association and resource
  allocation in full-duplex relay aided noma systems,'' \emph{IEEE Internet of
  Things Journal}, vol.~6, no.~6, pp. 10\,580--10\,596, 2019.

\bibitem{Yang2018Energy}
Z.~{Yang}, W.~{Xu}, Y.~{Pan}, C.~{Pan}, and M.~{Chen}, ``Energy efficient
  resource allocation in machine-to-machine communications with multiple access
  and energy harvesting for iot,'' \emph{IEEE Internet of Things Journal},
  vol.~5, no.~1, pp. 229--245, 2018.

\bibitem{Ansere2020Optimal}
J.~A. {Ansere}, G.~{Han}, L.~{Liu}, Y.~{Peng}, and M.~{Kamal}, ``Optimal
  resource allocation in energy-efficient internet-of-things networks with
  imperfect csi,'' \emph{IEEE Internet of Things Journal}, vol.~7, no.~6, pp.
  5401--5411, 2020.

\bibitem{chu2017wireless}
Z.~Chu, F.~Zhou, Z.~Zhu, R.~Q. Hu, and P.~Xiao, ``Wireless powered sensor
  networks for internet of things: Maximum throughput and optimal power
  allocation,'' \emph{IEEE Internet of Things Journal}, vol.~5, no.~1, pp.
  310--321, 2017.

\bibitem{Ngo_2013_Energy}
H.~Q. Ngo, E.~G. Larsson, and T.~L. Marzetta, ``Energy and spectral efficiency
  of very large multiuser {MIMO} systems,'' \emph{IEEE Transactions on
  Communications}, vol.~61, no.~4, pp. 1436--1449, Apr. 2013.

\bibitem{marzetta2016fundamentals}
T.~L. Marzetta, E.~G. Larsson, H.~Yang, and H.~Q. Ngo, \emph{Fundamentals of
  Massive MIMO}.\hskip 1em plus 0.5em minus 0.4em\relax Cambridge University
  Press, 2016.

\bibitem{bana2019massive}
A.-S. Bana, E.~De~Carvalho, B.~Soret, T.~Abr{\~a}o, J.~C. Marinello, E.~G.
  Larsson, and P.~Popovski, ``{Massive MIMO for Internet of Things (IoT)
  Connectivity},'' \emph{Physical Communication}, vol.~37, p. 100859, 2019.

\bibitem{Wang2020Wirelessly}
X.~{Wang}, A.~{Ashikhmin}, and X.~{Wang}, ``Wirelessly powered cell-free iot:
  Analysis and optimization,'' \emph{IEEE Internet of Things Journal}, vol.~7,
  no.~9, pp. 8384--8396, 2020.

\bibitem{Liu2018Massive1}
L.~{Liu} and W.~{Yu}, ``{Massive Connectivity With Massive MIMO—Part I:
  Device Activity Detection and Channel Estimation},'' \emph{IEEE Transactions
  on Signal Processing}, vol.~66, no.~11, pp. 2933--2946, 2018.

\bibitem{Liu2018Massive2}
------, ``{Massive Connectivity With Massive MIMO—Part II: Achievable Rate
  Characterization},'' \emph{IEEE Transactions on Signal Processing}, vol.~66,
  no.~11, pp. 2947--2959, 2018.

\bibitem{Ngo_17_Cellfree}
H.~Q. Ngo, A.~Ashikhmin, H.~Yang, E.~G. Larsson, and T.~L. Marzetta,
  ``Cell-free massive {MIMO} versus small cells,'' \emph{IEEE Transactions on
  Wireless Communications}, vol.~16, no.~3, pp. 1834--1850, Mar. 2017.

\bibitem{Hoydis_13_UL/DL}
J.~Hoydis, S.~ten Brink, and M.~Debbah, ``Massive {MIMO} in the {UL/DL} of
  cellular networks: How many antennas do we need?'' \emph{IEEE Journal on
  Selected Areas in Communications}, vol.~31, no.~2, pp. 160--171, Feb. 2013.

\bibitem{Wagner_12_MISO}
S.~{Wagner}, R.~{Couillet}, M.~{Debbah}, and D.~T.~M. {Slock}, ``Large system
  analysis of linear precoding in correlated {MISO} broadcast channels under
  limited feedback,'' \emph{IEEE Transactions on Information Theory}, vol.~58,
  no.~7, pp. 4509--4537, July 2012.

\bibitem{nayebi2016MMSE}
E.~{Nayebi}, A.~{Ashikhmin}, T.~L. {Marzetta}, and B.~D. {Rao}, ``{Performance
  of cell-free massive MIMO systems with MMSE and LSFD receivers},'' in
  \emph{2016 50th Asilomar Conference on Signals, Systems and Computers}, Nov.
  2016, pp. 203--207.

\bibitem{Tang_01_Pathloss}
A.~Tang, J.~Sun, and K.~Gong, ``Mobile propagation loss with a low base station
  antenna for {NLOS} street microcells in urban area,'' in \emph{IEEE 53rd
  Vehicular Technology Conference (VTC-Spring)}, 2001, pp. 333--336.

\bibitem{3GPP_ETSI}
3GPP, ``Digital cellular telecommunication system (phase 2+); radio network
  planning aspects,'' 3GPP, ETSI TR, 2010.

\bibitem{Wang08_Joint_Shadow}
Z.~Wang, E.~K. Tameh, and A.~R. Nix, ``Joint shadowing process in urban
  peer-to-peer radio channels,'' \emph{IEEE Transactions on Vehicular
  Technology}, vol.~57, no.~1, pp. 52--64, Jan. 2008.

\bibitem{Rao_Internet_2019}
S.~{Rao}, A.~{Ashikhmin}, and H.~{Yang}, ``{Internet of Things Based on
  Cell-Free Massive MIMO},'' in \emph{2019 53rd Asilomar Conference on Signals,
  Systems, and Computers}, 2019, pp. 1946--1950.

\bibitem{Rao_cellfree_2018}
------, ``{Cell-Free Massive MIMO with Nonorthogonal Pilots for Internet of
  Things},'' \emph{Bell Labs Report}, 2018.

\bibitem{boyd2004convex}
S.~Boyd and L.~Vandenberghe, \emph{Convex optimization}.\hskip 1em plus 0.5em
  minus 0.4em\relax Cambridge university press, 2004.

\bibitem{Nayebi2017Precoding}
E.~{Nayebi}, A.~{Ashikhmin}, T.~L. {Marzetta}, H.~{Yang}, and B.~D. {Rao},
  ``{Precoding and Power Optimization in Cell-Free Massive MIMO Systems},''
  \emph{IEEE Transactions on Wireless Communications}, vol.~16, no.~7, pp.
  4445--4459, July 2017.

\bibitem{levenberg1944method}
K.~Levenberg, ``A method for the solution of certain non-linear problems in
  least squares,'' \emph{Quarterly of applied mathematics}, vol.~2, no.~2, pp.
  164--168, 1944.

\bibitem{marquardt1963algorithm}
D.~W. Marquardt, ``An algorithm for least-squares estimation of nonlinear
  parameters,'' \emph{Journal of the society for Industrial and Applied
  Mathematics}, vol.~11, no.~2, pp. 431--441, 1963.

\bibitem{silverstein1995empirical}
J.~W. {Silverstein} and Z.~D. {Bai}, ``On the empirical distribution of
  eigenvalues of a class of large dimensional random matrices,'' \emph{Journal
  of Multivariate Analysis}, vol.~54, no.~2, pp. 175--192, 1995.

\bibitem{yates1995framework}
R.~D. {Yates}, ``A framework for uplink power control in cellular radio
  systems,'' \emph{IEEE Journal on Selected Areas in Communications}, vol.~13,
  no.~7, pp. 1341--1347, Sep. 1995.

\end{thebibliography}
\end{document}